\RequirePackage{amsmath}
\documentclass[runningheads]{llncs}
\usepackage[utf8]{inputenc}
\usepackage{xspace}
\usepackage{stmaryrd}
\usepackage{import}
\usepackage{graphicx}
\usepackage[caption=false,subrefformat=parens,labelformat=parens]{subfig}
\usepackage[linesnumbered,ruled,vlined]{algorithm2e}
\usepackage{comment}
\usepackage{adjustbox}
\usepackage{listings}
\usepackage{lstautogobble}
\usepackage{hyperref}
\usepackage{amssymb}
\usepackage{amsmath}
\usepackage{url}
\usepackage{xspace}
\usepackage{paralist}
\usepackage[dvipsnames]{xcolor}

\newcommand{\yv}[1]{{\color{red}[{\bf YV:} #1]}}
\newcommand{\ag}[1]{\textcolor{RoyalPurple}{[{\bf AG:} #1]}}

\newcommand{\limp}{\Rightarrow}

\newcommand{\Land}{\bigwedge}
\newcommand{\Lor}{\bigvee}

\newcommand{\sat}{\textsc{sat}\xspace}
\newcommand{\safe}{\textsc{safe}\xspace}
\newcommand{\unsafe}{\textsc{unsafe}}

\newcommand{\Tr}{\mathit{Tr}}
\newcommand{\Bad}{\mathit{Bad}}
\newcommand{\Init}{\mathit{Init}}
\newcommand{\Inv}{\mathit{Inv}}

\newcommand{\Pdr}{\textsc{Pdr}\xspace}

\newcommand{\cA}{\mathcal{A}}

\newcommand{\cL}{\mathcal{L}}
\newcommand{\cK}{\mathcal{K}}
\newcommand{\cW}{\mathcal{W}}

\newcommand{\Kavyextend}{\textsc{kAvyExtend}\xspace}
\newcommand{\Kavy}{\textsc{kAvy}\xspace}
\newcommand{\Kavyvanilla}{\textsc{vanilla}\xspace}

\newcommand{\PdrPush}{\textsc{PdrPush}\xspace}
\newcommand{\emptylist}{[\;]}
\newcommand{\vF}{\vec{F}}
\newcommand{\vH}{\vec{H}}
\newcommand{\vG}{\vec{G}}
\newcommand{\vA}{\vec{A}}
\newcommand{\vI}{\vec{I}}

\newcommand{\Avy}{\textsc{Avy}\xspace}

\newcommand{\AvyMkTrace}{\textsc{AvyMkTrace}\xspace}

\newcommand{\isSat}{\textsc{isSat}}
\newcommand{\seqItp}{\textsc{seqItp}}

\newcommand{\pdkind}{\textsc{Pd-Kind}\xspace}
\newcommand{\kicthree}{\textsc{KIC3}\xspace}
\newcommand{\ic}{\textsc{IC3}\xspace}

\newcommand{\trace}{[F_0,\ldots ,F_N] }

\newcommand{\Vars}{\bar{v}}
\newcommand{\Pdrblock }{\textsc{PdrBlock}\xspace}

\newcommand{\KindTr}[3]{\Tr\llbracket {#1}^{#2}\rrbracket^{#3}}
\newcommand{\Kineq}[2]{1\leq {#2} \leq {#1}+1 \leq N+1}
 \usepackage{wrapfig}
\usepackage[firstpage]{draftwatermark}
\SetWatermarkText{\hspace*{4.5in}\raisebox{8.5in}{\includegraphics[scale=0.1]{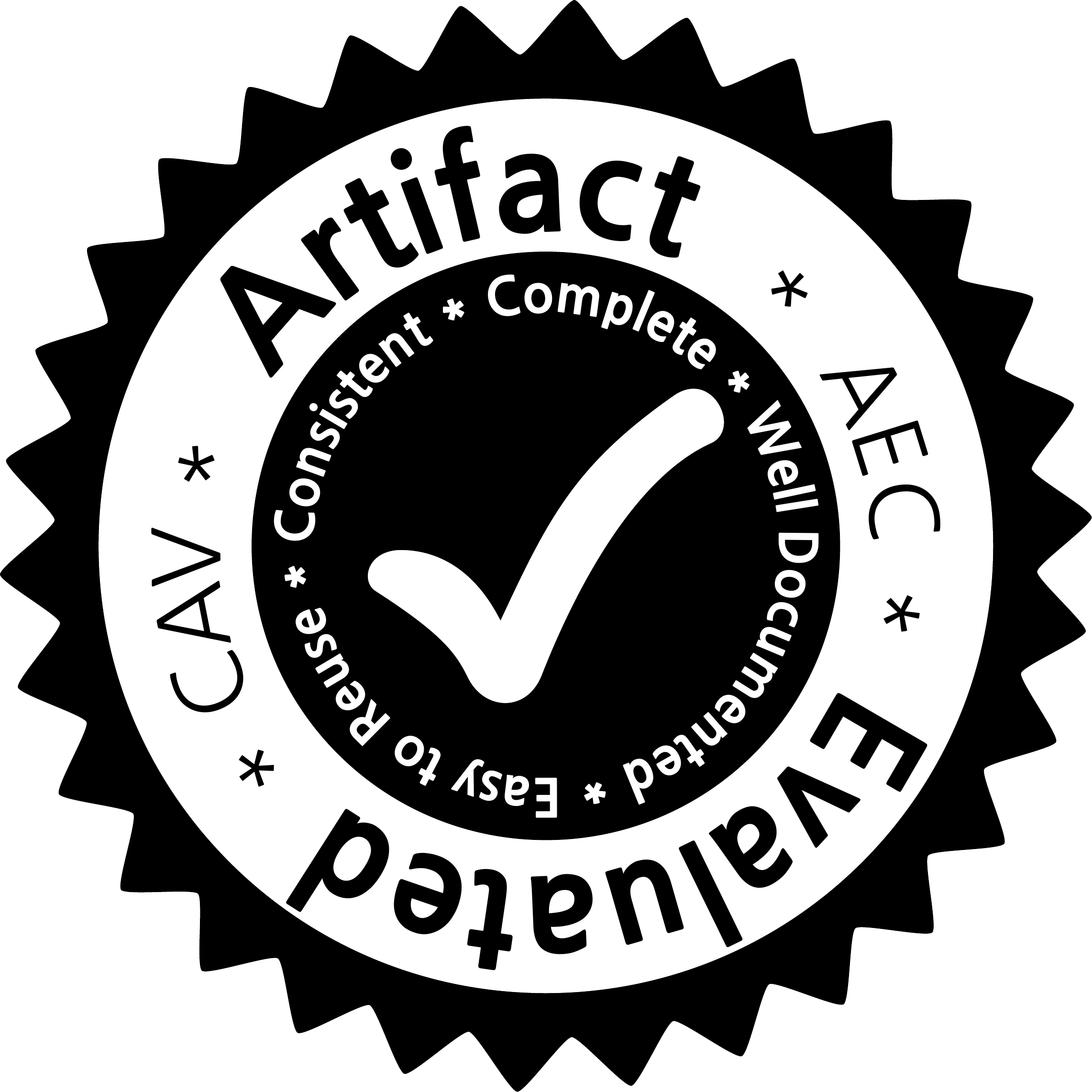}}}
\SetWatermarkAngle{0}
\title{Interpolating Strong Induction}
\author{Hari Govind V K\inst{1} \and Yakir Vizel\inst{2} \and Vijay
  Ganesh\inst{1} \and Arie Gurfinkel\inst{1}}
\institute{University of Waterloo \and The Technion}
\date{}
\begin{document}
\maketitle
\pagestyle{empty}
\DontPrintSemicolon
\vspace{-20pt}
\begin{abstract}
  The principle of strong induction, also known as $k$-induction is one of the
  first techniques for unbounded SAT-based Model Checking (SMC). While elegant
  and simple to apply, properties as such are rarely $k$-inductive and when they
  can be strengthened, there is no effective strategy to guess the depth of
  induction. It has been mostly displaced by techniques that compute inductive
  strengthenings based on interpolation and property directed reachability
  (\Pdr). In this paper, we present \Kavy, an SMC algorithm that effectively
  uses $k$-induction to guide interpolation and \Pdr -style inductive
  generalization. Unlike pure $k$-induction, \Kavy uses \Pdr -style
  generalization to compute and strengthen an inductive trace. Unlike pure \Pdr,
  \Kavy uses relative $k$-induction to construct an inductive invariant. The
  depth of induction is adjusted dynamically by minimizing a proof of
  unsatisfiability. We have implemented \Kavy within the \Avy Model Checker and
  evaluated it on HWMCC instances. Our results show that \Kavy is more effective
  than both \Avy and \Pdr, and that using $k$-induction leads to faster running
  time and solving more instances. Further, on a class of benchmarks, called
  {\it shift}, \Kavy is orders of magnitude faster than \Avy, \Pdr and
  $k$-induction.

\end{abstract}

 \section{Introduction}
\label{sec:intro}

The principle of strong induction, also known as $k$-induction, is a
generalization of (simple) induction that extends the base- and inductive-cases
to $k$ steps of a transition system~\cite{DBLP:conf/fmcad/SheeranSS00}. A safety
property $P$ is $k$-inductive in a transition system $T$ iff (a) $P$ is true in
the first $(k-1)$ steps of $T$, and (b) if $P$ is assumed to hold for $(k-1)$
consecutive steps, then $P$ holds in $k$ steps of $T$. Simple induction is
equivalent to $1$-induction. Unlike induction, strong induction is complete for
safety properties: a property $P$ is safe in a transition system $T$ iff there
exists a natural number $k$ such that $P$ is $k$-inductive in $T$ (assuming the
usual restriction to simple paths). This makes $k$-induction a powerful method
for unbounded SAT-based Model Checking (SMC).

Unlike other SMC techniques, strong induction reduces model checking to pure SAT
that does not require any additional features such as solving with
assumptions~\cite{DBLP:conf/fmcad/EenMA10},
interpolation~\cite{DBLP:conf/cav/McMillan03}, resolution
proofs~\cite{DBLP:conf/fmcad/HeuleHW13}, Maximal Unsatisfiable Subsets
(MUS)~\cite{DBLP:journals/jsat/BelovM12}, etc. It easily integrates with
existing SAT-solvers and immediately benefits from any improvements in
heuristics~\cite{DBLP:conf/sat/LiangOMTLG18,DBLP:conf/sat/LiangGPC16}, pre- and
in-processing~\cite{DBLP:conf/cade/JarvisaloHB12}, and parallel
solving~\cite{DBLP:conf/cp/AudemardLST16}. The simplicity of applying
$k$-induction made it the go-to technique for SMT-based infinite-state model
checking~\cite{DBLP:conf/cav/ChampionMST16,DBLP:conf/cav/MouraORRSST04,DBLP:conf/fmcad/JovanovicD16}.
In that context, it is particularly effective in combination with invariant
synthesis~\cite{DBLP:conf/nfm/KahsaiGT11,DBLP:conf/nfm/GarocheKT13}. Moreover,
for some theories, strong induction is strictly stronger than
$1$-induction~\cite{DBLP:conf/fmcad/JovanovicD16}: there are properties that
are $k$-inductive, but have no $1$-inductive strengthening.

Notwithstanding all of its advantages, strong induction has been mostly
displaced by more recent SMC techniques such as
Interpolation~\cite{DBLP:reference/mc/McMillan18}, Property Directed
Reachability~\cite{DBLP:conf/vmcai/Bradley11,DBLP:conf/fmcad/EenMB11,DBLP:conf/fmcad/GurfinkelI15,DBLP:conf/fmcad/BerryhillIVV17},
and their combinations~\cite{DBLP:conf/cav/VizelG14}. In SMC $k$-induction is
equivalent to induction: any $k$-inductive property $P$ can be strengthened to
an inductive property
$Q$~\cite{DBLP:conf/fmcad/GurfinkelI17,DBLP:conf/birthday/BjornerGMR15}. Even
though in the worst case $Q$ is exponentially
larger than $P$~\cite{DBLP:conf/birthday/BjornerGMR15}, this is rarely observed
in practice~\cite{DBLP:conf/fmcad/MebsoutT16}. Furthermore, the SAT queries get
very hard as $k$ increases and usually succeed only for rather small values of
$k$. A recent work~\cite{DBLP:conf/fmcad/GurfinkelI17} shows that strong
induction can be integrated in \Pdr. However,~\cite{DBLP:conf/fmcad/GurfinkelI17}
argues that $k$-induction is hard to control in the context of \Pdr since
choosing a proper value of $k$ is difficult. A wrong choice leads to a form of
state enumeration. In~\cite{DBLP:conf/fmcad/GurfinkelI17}, $k$ is fixed to $5$,
and regular induction is used as soon as $5$-induction fails.

In this paper, we present \Kavy, an SMC algorithm that effectively uses
$k$-induction to guide interpolation and \Pdr-style inductive generalization. As
many state-of-the-art SMC algorithms, \Kavy iteratively constructs candidate
inductive invariants for a given safety property $P$. However, the construction
of these candidates is driven by $k$-induction. Whenever $P$ is known to hold up
to a bound $N$, \Kavy searches for the smallest $k \leq N + 1$, such that either
$P$ or some of its strengthening is $k$-inductive. Once it finds the right $k$
and strengthening, it computes a $1$-inductive strengthening.

It is convenient to think of modern SMC algorithms (e.g., \Pdr and \Avy), and
$k$-induction, as two ends of a spectrum. On the one end, modern SMC algorithms
fix $k$ to $1$ and \emph{search} for a $1$-inductive strengthening of $P$. While
on the opposite end, $k$-induction fixes the strengthening of $P$ to be $P$
itself and \emph{searches} for a $k$ such that $P$ is $k$-inductive. \Kavy
\emph{dynamically} explores this spectrum, exploiting the interplay between
finding the right $k$ and finding the right strengthening.  

\begin{wrapfigure}{R}{.4\textwidth}
  \vspace{-0.5in}
\raggedleft
\begin{lstlisting}[language=Verilog,
  basicstyle=\linespread{0.5}\ttfamily\footnotesize,
  autogobble=true]
  reg [7:0] c = 0;
  always
    if(c == 64)
      c <= 0;
    else
      c <= c + 1; 
  end
  assert property (c < 66);
\end{lstlisting}
  \vspace{-10pt}
  \caption{ An example system. }
  \vspace{-20pt}
\label{fig:ex}
\end{wrapfigure}
As an example, consider a system in Fig.~\ref{fig:ex} that counts upto $64$ and
resets. The property, $p: c < 66$, is $2$-inductive. \ic, \Pdr and \Avy
iteratively guess a $1$-inductive strengthening of $p$. In the worst case, they
require at least $64$ iterations. On the other hand, \Kavy determines that $p$
is $2$-inductive after $2$ iterations, \emph{computes} a $1$-inductive invariant
$(c \neq 65) \land (c < 66)$, and terminates.

\Kavy builds upon the foundations of \Avy~\cite{DBLP:conf/cav/VizelG14}. \Avy
first uses Bounded Model Checking~\cite{DBLP:conf/tacas/BiereCCZ99} (BMC) to
prove that the property $P$ holds up to bound $N$. Then, it uses a sequence
interpolant~\cite{DBLP:conf/fmcad/VizelG09} and \Pdr-style
inductive-generalization~\cite{DBLP:conf/vmcai/Bradley11} to construct
$1$-inductive strengthening candidate for $P$. We emphasize that using
$k$-induction to construct $1$-inductive candidates allows \Kavy to efficiently
utilize many principles from \Pdr and \Avy. While maintaining $k$-inductive
candidates might seem attractive (since they may be smaller), they are also much
harder to generalize effectively~\cite{DBLP:conf/vmcai/Bradley11}.

We implemented \Kavy in the \Avy Model Checker, and evaluated it on the
benchmarks from the Hardware Model Checking Competition (HWMCC). Our experiments
show that \Kavy significantly improves the performance of \Avy and solves more
examples than either of \Pdr and \Avy. For a specific family of examples
from~\cite{DBLP:journals/mst/KovasznaiFB16}, \Kavy exhibits nearly constant time performance, compared to an
exponential growth of \Avy, \Pdr, and $k$-induction (see
Fig.~\ref{fig:shiftcactus} in Section~\ref{sec:evaluation}). This further
emphasizes the effectiveness of efficiently integrating strong induction into
modern SMC.

The rest of the paper is structured as follows. After describing the most
relevant related work, we present the necessary background in
Section~\ref{sec:background} and give an overview of SAT-based model checking
algorithms in Section~\ref{sec:satmc}. \Kavy is presented in
Section~\ref{sec:kavy}, followed by presentation of results in
Section~\ref{sec:evaluation}. Finally, we conclude the paper in
Section~\ref{sec:conclusion}. \paragraph{\textbf{Related work.}}
\Kavy builds on top of the ideas of \ic ~\cite{DBLP:conf/vmcai/Bradley11} and
\Pdr ~\cite{DBLP:conf/fmcad/EenMB11}. The use of interpolation for generating an
inductive trace is inspired by \Avy~\cite{DBLP:conf/cav/VizelG14}. While
conceptually, our algorithm is similar to \Avy, its proof of
correctness is non-trivial and is significantly different from that of
\Avy. We are not aware of any other work that combines interpolation with
strong induction.

There are two prior attempts enhancing \Pdr-style algorithms with $k$-induction.
\pdkind~\cite{DBLP:conf/fmcad/JovanovicD16} is an SMT-based Model Checking
algorithm for infinite-state systems inspired by \ic/\Pdr. It infers
$k$-inductive invariants driven by the property whereas \Kavy infers
$1$-inductive invariants driven by $k$-induction. \pdkind uses recursive
blocking with interpolation and model-based projection to block bad states, and
$k$-induction to propagate (push) lemmas to next level. While the algorithm is
very interesting it is hard to adapt it to SAT-based setting (i.e. SMC), and 
impossible to compare on HWMCC instances directly.

The closest related work is \kicthree~\cite{DBLP:conf/fmcad/GurfinkelI17}.
It modifies the counter example queue management strategy in \ic to
utilize $k$-induction during blocking. The main limitation is that the
value for $k$ must be chosen statically ($k=5$ is reported for the evaluation).
\Kavy also utilizes $k$-induction during blocking but computes the
value for $k$ dynamically. Unfortunately, the implementation is not available
publicly and we could not compare with it directly.

\section{Background}
\label{sec:background}

In this section, we present notations and background that is required
for the description of our algorithm.

\paragraph{Safety Verification.} A symbolic transition system $T$ is a tuple
$(\Vars, \Init, \Tr, \Bad)$, where $\Vars$ is a set of Boolean \emph{state}
variables. A state of the system is a complete valuation to all variables in
$\Vars$ (i.e., the set of states is $\{0,1\}^{|\Vars|}$). We write $\Vars' = \{v'
\mid v \in \Vars\}$) for the set of \emph{primed} variables, used to represent
the next state. $\Init$ and $\Bad$ are formulas over $\Vars$ denoting the set of
initial states and bad states, respectively, and $\Tr$ is a formula over $\Vars
\cup \Vars'$, denoting the transition relation. With abuse of notation, we use
formulas and the sets of states (or transitions) that they represent
interchangeably. In addition, we sometimes use a state $s$ to denote the formula
(cube) that characterizes it. For a formula $\varphi$ over $\Vars$, we use
$\varphi(\Vars')$, or $\varphi'$ in short, to denote the formula in which every
occurrence of $v \in \Vars$ is replaced by $v' \in \Vars'$. For simplicity of
presentation, we assume that the property $P=\neg\Bad$ is true in the initial
state, that is $\Init \limp P$.

Given a formula $\varphi(\Vars)$, an $M$-to-$N$-\emph{unrolling} of $T$, where
$\varphi$ holds in all intermediate states is defined by the formula:
\begin{equation}
\Tr[\varphi]_M^N = \Land_{i=M}^{N-1} \varphi(\Vars_i)\land\Tr(\Vars_i, \Vars_{i+1})
\end{equation}
We write $\Tr[\varphi]^N$ when $M=0$ and $\Tr_M^N$ when $\varphi = \top$.

A transition system $T$ is UNSAFE iff there exists a state $s\in\Bad$
s.t. $s$ is reachable, and is SAFE otherwise. Equivalently, $T$ is UNSAFE 
iff there exists a number $N$ such that the following \emph{unrolling} formula 
is satisfiable:
\begin{equation}\label{formula:bmc}
  \Init(\Vars_0) \land \Tr^N \land \Bad(\Vars_N)
\end{equation}
$T$ is SAFE if no such $N$ exists. Whenever $T$ is UNSAFE and $s_N\in\Bad$ is a
reachable state, the path from $s_0\in\Init$ to $s_N$ is called a
\emph{counterexample}.

An \emph{inductive invariant} is a formula $\Inv$ that satisfies: 
\begin{align}\label{eq:ind}
  \Init(\Vars) &\limp Inv(\Vars) & \Inv(\Vars) \land \Tr(\Vars,\Vars') &\limp \Inv(\Vars') 
\end{align}
A transition system $T$ is SAFE iff there exists an inductive
invariant $\Inv$ s.t. $Inv(\Vars)\limp P(\Vars)$. In this case we say that
$\Inv$ is a \emph{safe} inductive invariant.

The \emph{safety} verification problem is to decide whether a transition
system $T$ is SAFE or UNSAFE, i.e., whether there exists a safe inductive
invariant or a counterexample.

\paragraph{Strong Induction.} Strong induction (or $k$-induction) is a
generalization of the notion of an inductive invariant that is similar to how
``simple'' induction is generalized in mathematics. A formula $\Inv$ is
\emph{$k$-invariant} in a transition system $T$ if it is true in the first $k$
steps of $T$. That is, the following formula is valid:
  $\Init(\Vars_0)\land \Tr^k\limp\left(\Land_{i=0}^{k}\Inv(\Vars_i)\right)$.
A formula $\Inv$ is a \emph{$k$-inductive invariant}
iff $\Inv$ is a $(k-1)$-invariant and is inductive after $k$ steps of $T$, i.e., the
following formula is valid: $\Tr[\Inv]^k \limp \Inv(\Vars_{k})$. 
Compared to simple induction, $k$-induction strengthens the hypothesis in the
induction step: $\Inv$ is assumed to hold between steps $0$ to $k-1$ and is
established in step $k$. Whenever $\Inv\limp P$, we say that $\Inv$ is a safe
$k$-inductive invariant. An inductive invariant is a $1$-inductive
invariant.
\begin{theorem}
  \label{thm:k-one-induction}
  Given a transition system $T$. There exists a safe inductive invariant w.r.t.
  $T$ iff there exists a safe $k$-inductive invariant w.r.t. $T$.
\end{theorem}
Theorem~\ref{thm:k-one-induction} states that $k$-induction principle is as
complete as $1$-induction. One direction is trivial (since we can take $k=1$).
The other can be strengthened further: for every $k$-inductive invariant
$\Inv_k$ there exists a $1$-inductive strengthening $\Inv_1$ such that $\Inv_1
\limp \Inv_k$. Theoretically $\Inv_1$ might be exponentially bigger
than $\Inv_k$~\cite{DBLP:conf/birthday/BjornerGMR15}. In practice, both invariants tend to be of similar size.

We say that a formula $\varphi$ is \emph{$k$-inductive relative} to $F$ if it is a
$(k-1)$-invariant and $\Tr[\varphi\land F]^k \limp \varphi(\Vars_k)$. 

\paragraph{Craig Interpolation~\cite{DBLP:journals/jsyml/Craig57a}.}
We use an extension of Craig Interpolants to sequences, which is common in Model
Checking. Let $\vA=[A_1,\ldots,A_N]$ such that $A_1 \land \cdots \land A_N$ is
unsatisfiable. A \emph{sequence interpolant} $\vI = \seqItp(\vA)$ for $\vA$ is a
sequence of formulas $\vI = [I_2,\ldots,I_{N}]$ such that (a) $A_1 \limp I_2$,
(b) $\forall 1 < i < N \cdot I_{i} \land A_i \limp I_{i+1} $, (c) $I_{N} \land
A_N \limp \bot$, and (d) $I_i$ is over variables that are shared between the
corresponding prefix and suffix of $\vA$.

 \section{SAT-based Model Checking}
\label{sec:satmc}
In this section, we give a brief overview of SAT-based Model Checking
algorithms: \ic/\Pdr~\cite{DBLP:conf/vmcai/Bradley11,DBLP:conf/fmcad/EenMB11},
and \Avy~\cite{DBLP:conf/cav/VizelG14}.
While these algorithms are well-known, we give a uniform presentation and
establish notation necessary for the rest of the paper. We fix a symbolic transition system $T = (\Vars, \Init, \Tr, \Bad)$.

The main data-structure of these algorithms is a sequence of candidate
invariants, called an \emph{inductive trace}. An \emph{inductive trace}, or
simply a trace, is a sequence of formulas $\vF=\trace$ that satisfy the
following two properties:
\begin{align}
  \Init(\Vars) &= F_0(\Vars) & \forall 0 \leq i < N \cdot F_i(\Vars) \land
  \Tr(\Vars,\Vars') \limp F_{i+1}(\Vars')
\end{align}
An element $F_i$ of a trace is called a \emph{frame}. 
The index of a frame is called a \emph{level}. $\vF$ is \emph{clausal} when all 
its elements are in CNF. For convenience, we
view a frame as a set of clauses, and assume that a trace is padded with $\top$
until the required length. The \emph{size} of $\vF = \trace$ is $|\vF| = N$. For
$k\leq N$, we write $\vF^k=[F_k,\dots,F_N]$ for the $k$-suffix of $\vF$.

A trace $\vF$ of size $N$ is \emph{stronger} than a trace $\vG$ of size $M$ iff
$\forall 0 \leq i \leq \min(N,M) \cdot F_i(\Vars) \limp G_i(\Vars)$. A trace is
\emph{safe} if each $F_i$ is safe: $\forall i \cdot F_i \limp \neg \Bad$;
\emph{monotone} if $\forall 0 \leq i < N \cdot F_i \limp F_{i+1}$. In a monotone
trace, a frame $F_i$ over-approximates the set of states reachable in up to $i$
steps of the $\Tr$. A trace is closed if $\exists 1 \leq i \leq N
\cdot F_i \limp \left(\Lor_{j=0}^{i-1} F_j \right)$.

We define an unrolling formula of a $k$-suffix of a trace $\vF = [F_0,
\ldots, F_N]$ as
:
\begin{equation}
\Tr[\vF^k] = \Land_{i=k}^{|F|} F_i(\Vars_i)\land\Tr(\Vars_i, \Vars_{i+1})
\end{equation}
We write $\Tr[\vF]$ to denote an unrolling of a $0$-suffix of $\vF$ (i.e $\vF$ itself).
Intuitively, $\Tr[\vF^k]$ is satisfiable iff there is a $k$-step execution of
the $\Tr$ that is consistent with the $k$-suffix $\vF^k$. If a transition system
$T$ admits a safe trace $\vF$ of size $|\vF| = N$, then $T$ does not admit
counterexamples of length less than $N$. A safe trace $\vF$, with $|\vF| = N$ is
\emph{extendable} with respect to level $0\leq i\leq N$ iff there exists a safe
trace $\vG$ stronger than $\vF$ such that $|\vG| > N$ and $F_i \land \Tr
\limp G_{i+1}$. $\vG$ and the corresponding level $i$ are called an
\emph{extension trace} and an \emph{extension level} of $\vF$, respectively.
SAT-based model checking algorithms work by iteratively extending a given safe
trace $\vF$ of size $N$ to a safe trace of size~$N+1$.

An extension trace is not unique, but there is a largest extension level. We
denote the set of all extension levels of $\vF$ by $\cW(\vF)$. The
existence of an extension level $i$ implies that an unrolling of the $i$-suffix
does not contain any $\Bad$ states:
\begin{proposition}\label{prop:extend}
  Let $\vF$ be a safe trace. Then, $i$, $0 \leq i \leq N$, is an extension level of
  $\vF$ iff the formula $\Tr[\vF^i] \land \Bad(\Vars_{N+1})$ is unsatisfiable.
\end{proposition}

\begin{example}
For Fig.~\ref{fig:ex}, $\vF = [c = 0, c < 66]$ is a safe trace of size
$1$. The formula $(c < 66) \land \Tr \land \neg ( c' < 66 )$ is satisfiable.
Therefore, there does not exists an extension trace at level $1$. Since $( c = 0
) \land \Tr \land ( c' < 66 ) \land Tr' \land ( c'' \geq 66 )$ is unsatisfiable,
the trace is extendable at level $0$. For example, a valid extension trace at
level $0$ is $\vG = [c = 0, c < 2, c < 66]$.
\end{example}

Both \Pdr and \Avy iteratively extend a safe trace either until the
extension is closed or a counterexample is found. However, they differ in how
exactly the trace is extended.
In the rest of this section, we present \Avy and \Pdr through the lens of
extension level. The goal of
this presentation is to make the paper self-contained. We omit many important
optimization details, and refer the reader to the original
papers~\cite{DBLP:conf/fmcad/EenMB11,DBLP:conf/vmcai/Bradley11,DBLP:conf/cav/VizelG14}.

\Pdr maintains a monotone, clausal trace $\vF$ with $\Init$ as the first frame
($F_0$). The trace $\vF$ is extended by recursively computing and blocking (if
possible) states that can reach $\Bad$ (called \emph{bad states}). A bad state
is blocked at the largest level possible. Alg.~\ref{alg:pdr-block} shows
\Pdrblock, the backward search procedure that identifies and blocks bad states.
\Pdrblock maintains a queue of states and the levels at which they have to be
blocked. The smallest level at which blocking occurs is tracked in order to show
the construction of the extension trace. For each state $s$ in the queue, it is
checked whether $s$ can be blocked by the previous frame
$F_{d-1}$~(line~\ref{line:pred-check}). If not, a predecessor state $t$ of $s$
that satistisfies $F_{d-1}$ is computed and added to the
queue~(line~\ref{line:pred-add}). If a predecessor state is found at level $0$,
the trace is not extendable and an empty trace is returned. If the state $s$ is
blocked at level $d$, \textsc{PdrIndGen}, is called to generate a clause that
blocks $s$ and possibly others. The clause is then added to all the frames at
levels less than or equal to $d$. \textsc{PdrIndGen} is a crucial optimization
to \Pdr. However, we do not explain it for the sake of simplicity. The procedure
terminates whenever there are no more states to be blocked (or a counterexample 
was found at line~\ref{line:pdr-cex}). By construction, the output trace
$\vG$ is an extension trace of $\vF$ at the extension level $w$. Once \Pdr
extends its trace, $\PdrPush$ is called to check if the clauses it learnt
are also true at higher levels. \Pdr terminates when the trace is closed.
\begin{figure}[t]
  \begin{adjustbox}{scale=0.8}
\begin{minipage}[t]{7.7cm}
  \vspace{0pt}  
\begin{algorithm}[H]
    \label{alg:pdr-block}
    \caption{\Pdrblock.}
    \KwIn{A transition system $T=(\Init,\Tr,\Bad)$}
     \KwIn{A safe trace $\vF$ with $|\vF|=N$}
     \KwOut{An extension trace $\vG$ or an empty trace}
     $w \gets N+1$ ; $\vG \gets \vF$ ; $Q.push(\langle \Bad,N+1 \rangle)$\;
     \While{$\neg Q.empty()$}{
     $\langle s,d \rangle \gets  Q.pop()$\;
     \lIf{$d==0$}{
     \Return{$\emptylist$} \label{line:pdr-cex}}
     \If{$\isSat(  F_{d-1}(\Vars) \land  \Tr(\Vars,\Vars') \land s(\Vars'))$}
     {\label{line:pred-check}
     $t \gets \emph{predecessor}(s)$\;
     $Q.push(t,d-1)$\; \label{line:pred-add}
     $Q.push(s,d)$\;
     }
     \Else{
     $\forall 0 \leq i \leq d \cdot G_i \gets \left(  G_i \land
       \textsc{PdrIndGen}(\neg s) \right)$\;
     $w \gets min(w,d)$\;
     }
     }
     \Return{$\vG$}
\end{algorithm}
\end{minipage}
  \end{adjustbox}
\hfill
    \begin{adjustbox}{scale=0.8}
 \begin{minipage}[t]{7.6cm}
  \vspace{0pt}
  \begin{algorithm}[H]
    \label{alg:avy}
    \caption{$\Avy$.}
    \KwIn{A transition system $T = (\Init,\Tr,\Bad)$}
    \KwOut{\safe/\unsafe }
     $F_0\gets Init\mathbin{;}N\gets 0$\;
     \Repeat{$\infty$}{
     \lIf{$\isSat(\Tr[\vF^0] \land \Bad(\Vars_{N+1}))$}{\mbox{\Return{\unsafe}}}\label{line:avy-cex}
     $k \gets \max \{i \mid \neg \isSat(\Tr[\vF^i] \land \Bad(\Vars_{N+1}))\}$\;
     $I_{k+1},\ldots,I_{N+1} \gets \seqItp(\Tr[\vF^k] \land \Bad(\Vars_{N+1}))$\; \label{line:avy-seq-itp}
     $\forall 0 \leq i \leq k  \cdot G_i \gets F_i $\;\label{line:avy-g-construct-start} 
     $\forall k < i \leq (N+1) \cdot G_i \gets F_i \land I_i $\; \label{line:avy-g-construct}
     $\vF \gets \AvyMkTrace([G_0,\ldots,G_{N+1}])$\;\label{line:avy-mk-trace}
     $\vF \gets \PdrPush (\vF)$\;
     \lIf{$\exists 1\leq i \leq N \cdot F_i \limp \left( \Lor_{j=0}^{i-1}F_{j} \right)$}
     {\mbox{\Return{\safe}}}
     $N\gets N+1$\;
     }
\end{algorithm}
\end{minipage}
\end{adjustbox}
\vspace{-0.3in}
\end{figure}

Avy, shown in Alg.~\ref{alg:avy}, is an alternative to \Pdr
that combines interpolation and recursive
blocking. \Avy starts with a trace $\vF$, with $F_0 = \Init$, that is extended in every
iteration of the main loop. A counterexample is returned whenever $\vF$ is not
extendable~(line~\ref{line:avy-cex}). Otherwise, a sequence interpolant is
extracted from the unsatisfiability of $\Tr[\vF^{\max(\cW)}]
\land\Bad(\Vars_{N+1})$. A longer trace $\vG = [G_0, \ldots, G_N,G_{N+1}]$ is
constructed using the sequence interpolant~(line~\ref{line:avy-g-construct}).
Observe that $\vG$ is an extension trace of $\vF$. While $\vG$ is safe, it is
neither monotone nor clausal. A helper routine $\AvyMkTrace$ is used to convert
$\vG$ to a proper \Pdr trace on
line~\ref{line:avy-mk-trace}~(see~\cite{DBLP:conf/cav/VizelG14} for the details
on $\AvyMkTrace$). \Avy converges when the trace is closed.

 \section{Interpolating $k$-Induction}
\label{sec:kavy}
In this section, we present \Kavy, an SMC algorithm that uses the
principle of strong induction to extend an inductive trace. The section is
structured as follows. First, we introduce a concept of extending a trace using
relative $k$-induction. Second, we present \Kavy and describe the details of how
$k$-induction is used to compute an extended trace. Third, we describe two
techniques for computing maximal parameters to apply strong induction. Unless
stated otherwise, we assume that all traces are monotone.

A safe trace $\vF$, with $|\vF| = N$, is \emph{strongly extendable} with respect
to $(i,k)$, where $\Kineq{i}{k}$, iff there exists a safe inductive trace $\vG$
stronger than $\vF$ such that $|\vG| > N$ and $\Tr[F_i]^k \limp G_{i+1}$. We
refer to the pair $(i, k)$ as \emph{a strong extension level (SEL)}, and to the
trace $\vG$ as an \emph{$(i,k)$-extension trace}, or simply a \emph{strong
  extension trace (SET)} when $(i, k)$ is not important. Note that for $k=1$,
$\vG$ is just an extension trace. 

\begin{example}\label{example:sel}
For Fig.~\ref{fig:ex}, the trace $\vF = [c = 0, c < 66]$ is
strongly extendable at level $1$. A valid $(1, 2)$-externsion trace is $ \vG =
[c = 0, ( c \neq 65 ) \land ( c < 66 ), c < 66]$. 
Note that $( c < 66 )$ is $2$-inductive relative to $F_1$, 
i.e. $\Tr[F_1]^2\limp ( c'' < 66 )$.
\vspace{-10pt}
\end{example}

We write $\cK(\vF)$ for the set of all SELs of $\vF$. We define an order on SELs
by : $(i_1,k_1) \preceq (i_2, k_2)$
iff \begin{inparaenum}[(i)] \item $i_1 < i_2$; or \item $i_1 = i_2\land k_1 >
  k_2$.
\end{inparaenum} The maximal SEL is $\text{max}(\cK(\vF))$.

Note that the existence of a SEL $(i,k)$ means that an unrolling of the $i$-suffix
with $F_i$ repeated $k$ times does not contain any bad states. We use $\KindTr{\vF}{i}{k}$ to denote this \emph{characteristic formula} for SEL $(i,k)$ : 
\begin{equation}
\KindTr{\vF}{i}{k} = \begin{cases}
  \Tr[F_i]_{i+1-k}^{i+1}\land\Tr[\vF^{i+1}] & \text{if } 0 \leq i < N   \\
  \Tr[F_N]_{N+1-k}^{N+1} & \text{if } i = N  
  \end{cases}
\end{equation}

\begin{proposition}\label{prop:ind_extend}
  Let $\vF$ be a safe trace, where $|\vF| = N$. Then, $(i,k)$, $\Kineq{i}{k}$, is an SEL of $\vF$ iff the formula
$\KindTr{\vF}{i}{k}\land\Bad(\Vars_{N+1})$ is
  unsatisfiable.

\end{proposition}

The level $i$ in the maximal SEL $(i,k)$ of a given trace $\vF$ is greater or
equal to the maximal extension level of $\vF$:
\begin{lemma}
  \label{lem:kavy_to_avy}
   Let $(i, k) = \max(\cK(\vF))$, then $i  \geq \max(\cW(\vF))$. 
\end{lemma}
Hence, extensions based on maximal SEL are constructed from frames at higher
level compared to extensions based on maximal extension level.

\begin{example}
For Fig.~\ref{fig:ex}, the trace $[c = 0, c < 66]$ has a maximum
extension level of $0$. Since $(c < 66)$ is $2$-inductive, the trace is strongly
extendable at level $1$ (as was seen in Example~\ref{example:sel}).
\end{example}

\subsubsection{\Kavy Algorithm}\Kavy is shown in Fig.~\ref{alg:kavy}. It starts
with an inductive trace $\vF=[Init]$ and iteratively extends $\vF$ using
SELs. A counterexample is returned if the trace cannot be
extended~(line~\ref{line:kavy_cex}). Otherwise, \Kavy computes the largest
extension level (line~\ref{line:best_w})~(described in Section~\ref{sec:kiew}).
Then, it constructs a strong extension trace using \Kavyextend
(line~\ref{line:kavyextend})~(described in Section~\ref{sec:kavyextend}).
Finally, \PdrPush is called to check whether the trace is closed. Note
that $\vF$ is a monotone, clausal, safe inductive trace throughout the algorithm.

\setlength{\textfloatsep}{10pt}
\begin{algorithm}[t]
    \label{alg:kavy}
    \caption{\Kavy algorithm.}
    \KwIn{A transition system $T = (\Init,\Tr,\Bad)$}
    \KwOut{\safe/\unsafe }
     $\vF\gets [Init]\mathbin{;}N\gets 0$\;
     \Repeat{$\infty$}{
     \tcp{Invariant: $\vF$ is a monotone, clausal, safe, inductive trace}
     $U \gets \Tr[\vF^0] \land \Bad(\Vars_{N+1})  $\;\label{line:bmc}
     \lIf{$\isSat(U)$}{\Return{\unsafe}\label{line:kavy_cex}}
     $(i,k) \gets \max\{(i, k) \mid \neg\isSat(\KindTr{\vF}{i}{k} \land \Bad(\Vars_{N+1}))\} $\label{line:best_w}\;
     $[F_0,\ldots,F_{N+1}] \gets \Kavyextend(\vF,(i,k))$\;\label{line:kavyextend}
$[F_0,\ldots,F_{N+1}] \gets \PdrPush ([F_0,\ldots,F_{N+1}])$\;
     \lIf{$\exists 1\leq i \leq N \cdot F_i \limp \left( \Lor_{j=0}^{i-1}F_{j} \right)$}
     {\Return{\safe}}
     $N\gets N+1$\;
     }
\end{algorithm}
\subsection{Extending a Trace with Strong Induction}
\label{sec:kavyextend}

In this section, we describe the procedure \Kavyextend (shown in
Alg.~\ref{alg:kavyextend}) that given a trace $\vF$ of size $|\vF| = N$ and an
$(i,k)$ SEL of $\vF$ constructs an $(i,k)$-extension trace $\vG$ of size $|\vG|
= N + 1$. The procedure itself is fairly simple, but its proof of correctness is
complex. We first present the theoretical results that connect
sequence interpolants with strong extension traces, then the procedure, and then
details of its correctness. Through the section, we fix a trace $\vF$ and its
SEL $(i,k)$.

\paragraph{Sequence interpolation for SEL.}
Let $(i,k)$ be an SEL of $\vF$. By Proposition~\ref{prop:ind_extend}, $\Psi =
\KindTr{\vF}{i}{k}\land\Bad(\Vars_{N+1})$ is unsatisfiable. Let $\cA =
\{A_{i-k+1}, \dots, A_{N+1}\}$ be a partitioning of $\Psi$ defined as follows:
\[
  A_j =
  \begin{cases}
    F_i(\Vars_{j})\land\Tr(\Vars_j,\Vars_{j+1}) & \text{if } i-k+1\leq j \leq i\\
    F_j(\Vars_{j})\land\Tr(\Vars_j,\Vars_{j+1}) & \text{if } i< j \leq N \\
   \Bad(\Vars_{N+1}) & \text{if } j = N+1 
  \end{cases}
\]
Since $( \land \cA ) = \Psi$, $\cA$ is unsatisfiable. Let
$\vI=[I_{i-k+2},\ldots,I_{N+1}]$ be a sequence interpolant corresponding to
$\cA$. Then, $\vI$ satisfies the following properties:
\begin{align}
  \label{eq:itp_prop}
F_i\land\Tr &\limp I'_{i-k+2} & 
\forall i-k+2 \leq j \leq i\cdot(F_i\land I_{j})\land\Tr &\limp I'_{j+1}\tag{$\heartsuit$}\\
I_{N+1}&\limp\neg\Bad &
\forall i < j \leq N\cdot (F_j\land I_j)\land\Tr &\limp I'_{j+1}\tag*{} 
\end{align}
Note that in~\eqref{eq:itp_prop}, both $i$ and $k$ are
fixed --- they are the $(i,k)$-extension level. Furthermore, in the top row
$F_i$ is fixed as well.

The conjunction of the first $k$ interpolants in $\vI$ is $k$-inductive
relative to the frame $F_i$:
\begin{lemma}\label{lem:seqitp_kind}
  The formula $F_{i+1}\land\left(\Land\limits_{m=i-k+2}^{i+1}
    I_m\right)$ is $k$-inductive relative to $F_i$.
\end{lemma}
\begin{proof}
  Since $F_i$ and $F_{i+1}$ are consecutive frames of a trace, $F_i\land\Tr\limp
  F_{i+1}'$. Thus, $\forall i-k+2\leq j\leq i\cdot \Tr[F_i]_{i-k+2}^{j}\limp
  F_{i+1}(\Vars_{j+1})$. Moreover, by~\eqref{eq:itp_prop}, $F_i\land\Tr\limp
  I_{i-k+2}'$ and $\forall i-k+2\leq j\leq i+1\cdot (F_i\land I_j)\land \Tr\limp
  I_{j+1}'$. Equivalently, $\forall i-k+2\leq j\leq i+1\cdot
  \Tr[F_i]_{i-k+2}^{j}\limp I_{j+1}(\Vars_{j+1})$. By induction over the
  difference between $(i+1)$ and $(i-k+2)$, we show that
  $\Tr[F_i]_{i-k+2}^{i+1}\limp
  (F_{i+1}\land\Land_{m=i-k+2}^{i+1}I_{m})(\Vars_{i+1})$, which concludes the
  proof.\qed
  \end{proof}

  We use Lemma~\ref{lem:seqitp_kind} to define a strong extension trace $\vG$:
\begin{lemma}\label{lem:seqitp_con}
Let $\vG = [G_0,\ldots,G_{N+1}]$, be an inductive trace defined as follows:  
\[
  G_j =
  \begin{cases}
    F_j & \text{if } 0\leq j < i-k+2 \\
   F_{j}\land\left(\Land\limits_{m=i-k+2}^{j} I_m\right)  & \text{if } i-k+2\leq
   j < i+2 \\
   (F_{j}\land I_{j}) &\text{if }i+2\leq j< N+1 \\
  I_{N+1} & \text{if j = (N+1)} \\
  \end{cases}
\]
Then, $\vG$ is an $(i,k)$-extension trace of $\vF$ (not necessarily monotone).
\end{lemma}
\begin{proof}
  By Lemma~\ref{lem:seqitp_kind}, $G_{i+1}$ is $k$-inductive relative to $F_i$.
  Therefore, it is sufficient to show that $\vG$ is a safe inductive trace that is
  stronger than $\vF$. By definition, $\forall 0\leq j \leq N\cdot
  G_j\limp F_j$. By~\eqref{eq:itp_prop}, 
  $F_i\land\Tr\limp I_{i-k+2}'$ and $\forall i-k+2\leq j < i+2\cdot (F_i\land
  I_j)\land \Tr\limp I_{j+1}'$. By induction over $j$, 
  $\left(  (F_i\land \Land_{m=i-k+2}^{j}I_m)\land \Tr\right)\limp \Land_{m=i-k+2}^{j+1}I_{m}'$
  for all $i-k+2\leq j < i + 2$. Since $\vF$ is monotone, $\forall i - k + 2 \leq j <
  i+2 \cdot \left(  (F_j\land \Land_{m=i-k+2}^{j}I_m)\land \Tr\right)\limp \Land_{m=i-k+2}^{j+1}I_{m}'$

  By~\eqref{eq:itp_prop}, $\forall i < j \leq N\cdot (F_j\land I_j)\land\Tr
  \limp I_{j+1}'$. Again, since $\vF$ is a trace, we conclude that $\forall i <
  j < N\cdot (F_j\land I_j)\land\Tr \limp (F_{j+1}\land I_{j+1})'$. Combining
  the above, $G_j\land\Tr\limp G_{j+1}'$ for $0\leq j\leq N$. Since $\vF$ is
  safe and $I_{N+1}\limp\neg\Bad$, then $\vG$ is safe and stronger than $\vF$.
\qed
\end{proof}

Lemma~\ref{lem:seqitp_con} defines an obvious procedure to construct an
$(i,k)$-extension trace $\vG$ for $\vF$. However, such $\vG$ is neither monotone
nor clausal. In the rest of this section, we describe the procedure \Kavyextend
that starts with a sequence interpolant (as in Lemma~\ref{lem:seqitp_con}), but
uses \Pdrblock to systematically construct a safe monotone clausal extension of
$\vF$. 

\begin{algorithm}[t]
\label{alg:kavyextend}
\caption{\Kavyextend. The invariants marked $^\dagger$ hold only when the \Pdrblock does no inductive generalization.}
\KwIn{a monotone, clausal, safe trace $\vF$  of size $N$}
\KwIn{A strong extension level $(i,k)$ s.t. $\KindTr{\vF}{i}{k} \land \Bad(\Vars_{N+1})$ is unsatisfiable}
\KwOut{a monotone, clausal, safe trace $\vG$ of size $N+1$}
$I_{i-k+2},\ldots,I_{N+1 } \gets \seqItp(\KindTr{\vF}{i}{k} \land \Bad(\Vars_{N+1}))$\;\label{line:itp}

$\vG \gets [F_{0},\ldots, F_N,\top]$\;\label{line:init_g}

\For{$j \gets i - k + 1$ \KwTo $i$ }
{\label{line:g0}
$P_j \gets ( G_{j}\lor (G_{i+1}\land I_{j+1}) )$\;
\tcp{$\text{Inv}_1$: $\vG$ is monotone and clausal}
\tcp{$\text{Inv}_2$: $G_{i} \land Tr \limp P_j$}
\tcp{$\text{Inv}^\dagger_3$ : $\forall j < m\leq (i+1) \cdot G_{m} \equiv F_{m} \land \Land_{\ell=i-k+1}^{j-1} \left( G_{\ell} \lor I_{\ell+1}  \right)$}
\tcp{$\text{Inv}_3$ : $\forall j < m\leq (i+1) \cdot G_{m} \limp F_{m} \land \Land_{\ell=i-k+1}^{j-1} \left( G_{\ell} \lor I_{\ell+1}  \right)$}
$[ \_,\_, G_{i+1} ] \gets \Pdrblock( [ Init, G_i, G_{i+1} ], (\Init,\Tr, \neg P_j ))$\; \label{line:firstLoop}
}\label{line:g0_end}
$P_{i} \gets ( G_{i}\lor (G_{i+1}\land I_{j+1}) )$\;
\lIf{$i=0$}{
$[ \_,\_, G_{i+1} ] \gets \Pdrblock( [ \Init, G_{i+1} ], (\Init,\Tr,\neg P_i))$
}
\lElse
{
$[ \_,\_, G_{i+1} ] \gets \Pdrblock( [ \Init, G_{i}, G_{i+1} ], (\Init,\Tr, \neg P_i ))$
}\label{line:g-k-ind}
\tcp{$\text{Inv}^\dagger_4$: $G_{i+1} \equiv F_{i+1} \land \Land_{\ell=i-k+1}^{i} \left( G_\ell \lor I_{\ell+1} \right)$}
\tcp{$\text{Inv}_4$: $G_{i+1} \limp F_{i+1} \land \Land_{\ell=i-k+1}^{i} \left( G_\ell \lor I_{\ell+1} \right)$}
\For{$j \gets i + 1$ \KwTo $N + 1$}
{\label{line:g1}
$P_j \gets G_j \lor (G_{j+1} \land I_{j+1})$\;
  \tcp{$\text{Inv}_6$: $G_j \land \Tr \limp P_j$}
$[ \_,\_, G_{j+1} ] \gets \Pdrblock( [ \Init, G_{j}, G_{j+1} ], (\Init,\Tr,\neg P_j ) )$\; \label{line:secondLoop}
$\vG \gets \PdrPush(\vG)$\; \label{line:kavy-pdrpush}
}\label{line:g1_end}
\tcp{$\text{Inv}^\dagger_7$: $\vG$ is an $(i,k)$-extension trace of $\vF$}
\tcp{$\text{Inv}_7$: $\vG$ is an extension trace of $\vF$}
\Return{$\vG$}
\end{algorithm}

The procedure \Kavyextend is shown in Alg.~\ref{alg:kavyextend}. For simplicity
of the presentation, we assume that $\Pdrblock$ does not use inductive generalization. The invariants
marked by $^\dagger$ rely on this assumption. We stress that the assumption is
for presentation only. The correctness of \Kavyextend is independent~of~it.

\Kavyextend starts with a sequence interpolant according to the partitioning
$\cA$. The extension trace $\vG$ is initialized to $\vF$ and $G_{N+1}$ is
initialized to $\top$~(line~\ref{line:init_g}). The rest proceeds in three phases:
\emph{Phase~1} (lines~\ref{line:g0}--\ref{line:g0_end}) computes the prefix
$G_{i-k+2},\ldots,G_{i+1}$ using the first $k-1$ elements of $\vI$;
\emph{Phase~2} (line~\ref{line:g-k-ind}) computes $G_{i+1}$ using $I_{i+1}$;
\emph{Phase~3} (lines~\ref{line:g1}--\ref{line:g1_end}) computes the suffix
$\vG^{i+2}$ using the last $( N-i )$ elements of $\vI$. During this phase, $\PdrPush$~(line~{\ref{line:kavy-pdrpush}}) pushes clauses forward so that they can be used in the next iteration. The correctness of the
phases follows from the invariants shown in Alg.~\ref{alg:kavyextend}. We
present each phase in turn.

Recall that \Pdrblock takes a trace $\vF$ (that is safe up to the last frame)
and a transition system, and returns a safe strengthening of $\vF$, while
ensuring that the result is monotone and clausal. This guarantee is maintained
by Alg~\ref{alg:kavyextend}, by requiring that any clause added to any frame
$G_i$ of $\vG$ is implicitly added to all frames below $G_i$.

\paragraph{Phase~1.} By Lemma~\ref{lem:seqitp_kind}, the first $k$ elements of
the sequence interpolant  computed at
line~\ref{line:itp} over-approximate states reachable in $i+1$ steps of $\Tr$.
Phase~1 uses this to strengthen $G_{i+1}$ using the first $k$ elements of
$\vI$. Note that in that phase, new clauses are always added to frame $G_{i+1}$,
and all frames before it! 

Correctness of Phase~1 (line~\ref{line:firstLoop}) follows from the loop
invariant $\texttt{Inv}_2$. It holds on loop entry since $G_i \land \Tr \limp
I_{i-k+2}$ (since $G_i = F_i$ and~\eqref{eq:itp_prop}) and $G_i \land \Tr \limp G_{i+1}$ (since $\vG$ is initially a trace). Let $G_i$ and $G_i^*$ be the
$i^{th}$ frame before and after execution of iteration $j$ of the loop,
respectively. \Pdrblock blocks $\neg P_j$ at iteration $j$ of the loop. Assume that $\texttt{Inv}_2$ holds at the beginning of the loop.
Then, $G_i^* \limp G_i \land P_j$ since
\Pdrblock strengthens $G_i$. Since $G_j \limp G_i$ and $G_i \limp
G_{i+1}$, this  
simplifies to $G_i^* \limp G_j \lor ( G_i \land I_{j+1} )$. Finally, since $\vG$
is a trace, $\texttt{Inv}_2$ holds at the end of the iteration.

$\texttt{Inv}_2$ ensures that the trace given to \Pdrblock at
line~\ref{line:firstLoop} \emph{can} be made safe relative to $P_j$. From the post-condition of
\Pdrblock, it follows that at iteration $j$, $G_{i+1}$ is strengthened to
$G_{i+1}^*$ such that $G_{i+1}^* \limp P_j$ and
$\vG$ remains a monotone clausal trace. At the end of \emph{Phase~1},
$[G_0,\ldots,G_{i+1}]$ is a clausal monotone trace. 

 Interestingly, the calls to \Pdrblock in this phase do not
satisfy an expected pre-condition: the frame $G_i$ in $[\Init,G_i,G_{i+1}] $
might not be safe for property $P_j$. However, we can see that $\Init \limp P_j$ and from $\texttt{Inv}_2$, it is clear that $P_j$ is inductive relative to $G_i$. 
This is a sufficient precondition for \Pdrblock .

\paragraph{Phase~2.} This phase strengthens $G_{i+1}$ using the interpolant
$I_{i+1}$. After Phase~2, $G_{i+1}$ is $k$-inductive relative to $F_i$.

\paragraph{Phase~3.} Unlike \emph{Phase~1}, $G_{j+1}$ is computed at the
$j^{th}$ iteration. Because of this, the property $P_j$ in this phase is
slightly different than that of Phase~1. Correctness follows from invariant
$\texttt{Inv}_6$ that ensures that at iteration $j$, $G_{j+1}$ \emph{can} be
made safe relative to $P_j$. From the post-condition of \Pdrblock, it follows
that $G_{j+1}$ is strengthened to $G_{j+1}^*$ such that $G_{j+1}^* \limp P_j$
and $\vG$ is a monotone clausal trace. The invariant implies that at the end of
the loop $G_{N+1} \limp G_{N} \lor I_{N+1}$, making $\vG$ safe. Thus, at the end
of the loop $\vG$ is a safe monotone clausal trace that is stronger than
$\vF$. What remains is to show is that $G_{i+1}$ is $k$-inductive relative to
$F_i$.

Let $\varphi$ be the formula from Lemma~\ref{lem:seqitp_kind}. Assuming that
$\Pdrblock$ did no inductive generalization, \emph{Phase~1} maintains
$\texttt{Inv}^\dagger_3$, which states that at iteration $j$,
\Pdrblock strengthens frames $\{G_m\}$, $j < m \leq (i+1)$.
$\texttt{Inv}^\dagger_3$ holds on loop entry, since initially $\vG=\vF$. Let
$G_{m}$, $G_{m}^*$ ( $j < m \leq (i+1)$ ) be frame $m$ at the beginning and at
the end of the loop iteration, respectively. In the loop, \Pdrblock adds clauses
that block $\neg P_j$. Thus, $G_{m}^* \equiv G_{m} \land P_j $. Since $G_j \limp
G_{m}$, this simplifies to $G^*_{m} \equiv G_{m} \land (G_{j} \lor I_{j+1})$.
Expanding $G_{m}$, we get $G^*_{m} \equiv F_{m} \land \Land_{\ell=i-k+1}^{j}
\left( G_\ell \lor I_{\ell+1} \right)$. Thus, $\texttt{Inv}^\dagger_3$ holds at
the end of the loop.

In particular, after line~\ref{line:g-k-ind}, 
$G_{i+1} \equiv F_{i+1} \land \Land_{\ell=i-k+1}^{i} \left( G_\ell \lor I_{\ell+1} \right)$.
Since $\varphi\limp G_{i+1}$, $G_{i+1}$ is $k$-inductive relative to $F_i$.

\begin{theorem}
  Given a safe trace $\vF$ of size $N$ and an SEL $(i,k)$ for $\vF$, \Kavyextend
  returns a clausal monotone extension trace $\vG$ of size $N+1$.
  Furthermore, if $\Pdrblock$ does no inductive generalization then $\vG$ is
  an $(i,k)$-extension trace.
\end{theorem}

Of course, assuming that \Pdrblock does no inductive generalization is not
realistic. \Kavyextend remains correct without the assumption: it returns a
trace $\vG$ that is a monotone clausal extension of $\vF$. However, $\vG$ might be
stronger than any $(i,k)$-extension of $\vF$. The invariants marked with
$^\dagger$ are then relaxed to their unmarked versions. Overall,
inductive generalization improves \Kavyextend since it is not restricted to only
a $k$-inductive strengthening.

Importantly, the output of \Kavyextend is a regular inductive trace. Thus,
\Kavyextend is a procedure to strengthen a (relatively) $k$-inductive
certificate to a (relatively) $1$-inductive certificate. Hence, after
\Kavyextend, any strategy for further generalization or trace extension from
\ic, \Pdr, or \Avy is applicable.

 \subsection{Searching for the maximal SEL}
\label{sec:kiew}

In this section, we describe two algorithms for computing the maximal SEL. Both
algorithms can be used to implement line~\ref{line:best_w} of
Alg.~\ref{alg:kavy}. They perform a guided search for group minimal
unsatisfiable subsets. They terminate when having fewer clauses would not
increase the SEL further. The first, called \emph{top-down}, starts from the
largest unrolling of the $\Tr$ and then reduces the length of the unrolling. The
second, called \emph{bottom-up}, finds the largest (regular) extension level
first, and then grows it using strong induction.

\paragraph{Top-down SEL.} A
pair $(i,k)$ is the maximal SEL iff
\begin{align*}
i &= \max\; \{j  \mid 0 \leq j \leq N \cdot \KindTr{\vF}{j}{j+1}\land \Bad(\Vars_{N+1}) \limp \bot\} \\
k &= \min\; \{ \ell \mid 1 \leq \ell \leq (i+1) \cdot \KindTr{\vF}{i}{\ell}\land \Bad(\Vars_{N+1}) \limp \bot\}
\end{align*}
Note that $k$ depends on $i$. For a SEL $(i,k)\in \cK(\vF)$, we refer to the formula
$\Tr[\vF^i]$ as a \emph{suffix} and to number $k$ as the depth of induction. Thus, the search can be split into two phases: (a)
find the smallest suffix while using the maximal depth of induction
allowed (for that suffix), and (b) minimizing the depth of induction $k$ for the
value of $i$ found in step (a). This is captured in Alg. \ref{alg:top-down}. The
algorithm requires at most $( N+1 )$ \sat queries. One downside, however, is
that the formulas constructed in the first phase (line~\ref{line:big-formula}) are large because the
depth of induction is the maximum possible.

\begin{figure}[t]
  \begin{adjustbox}{scale=0.9}
\begin{minipage}[t]{7cm}
  \vspace{0pt}  
\begin{algorithm}[H]
\label{alg:top-down}
\caption{A top down alg. for the maximal SEL.}
\KwIn{A transition system $T=(\Init,\Tr,\Bad)$}
\KwIn{An extendable monotone clausal safe trace $\vF$ of size $N$}
\KwOut{$\max(\cK(\vF))$}
$i \gets N $\;
\While{$i > 0$ }
{
\lIf{$\neg \isSat(\KindTr{\vF}{i}{i+1}\land \Bad(\Vars_{N+1}))$}
{\label{line:big-formula}
\textbf{break}
}
$i\gets (i-1)$\;
}
$k \gets 1$\;
\While{$k < i + 1$ }
{
\lIf{$\neg \isSat(\KindTr{\vF}{i}{k}\land \Bad(\Vars_{N+1}))$}
{
\textbf{break}
}
$k \gets ( k+1 )$\;
}
\Return{$(i,k)$}
\end{algorithm}
\end{minipage}
  \end{adjustbox}
\hfill
    \begin{adjustbox}{scale=0.9}
\begin{minipage}[t]{6.5cm}
  \vspace{0pt}
\begin{algorithm}[H]
\label{alg:bot-up}
\caption{A bottom up alg. for the maximal SEL.}
\KwIn{A transition system $T=(\Init,\Tr,\Bad)$}
\KwIn{An extendable monotone clausal safe trace $\vF$ of size $N$}
\KwOut{$\max(\cK(\vF))$}
$j \gets N $\;
\While{$j > 0$ }
{\label{line:min-suffix}
\lIf{$\neg \isSat(\KindTr{\vF}{j}{1}\land \Bad(\Vars_{N+1}))$}
{
\textbf{break}
}
$j\gets (j-1)$\;
}
$(i,k) \gets (j,1) \mathbin{;} j \gets (j+1)\mathbin{;}\ell \gets 2$\;\label{line:first-sel}
\While{$\ell \leq (j+1) \land j \leq N $}
{\label{line:out-while}
\lIf{$\isSat(\KindTr{\vF}{j}{\ell}\land \Bad(\Vars_{N+1}))$}
{\label{line:k-incr}
$\ell \gets (\ell +1)$
}
\Else
{
$(i,k)\gets (j,\ell)$\;\label{line:sel}
$j \gets (j +1)$\;\label{line:i-incr}
}
\label{line:end-out-while}
}
\Return{$(i,k)$}
\end{algorithm}
\end{minipage}
\end{adjustbox}
\vspace{-10pt}
\end{figure}

\paragraph{Bottom-up SEL.} 
Alg.~\ref{alg:bot-up} searches for a SEL by first finding a maximal regular
extension level~(line~\ref{line:min-suffix}) and then searching for larger
SELs~(lines~\ref{line:out-while}~to~\ref{line:end-out-while}). Observe that if
$(j,\ell) \not \in \cK(\vF)$, then $\forall p > j \cdot (p,\ell)\not \in
\cK(\vF)$. This is used at line~\ref{line:k-incr} to increase the depth of
induction once it is known that $(j,\ell)\not \in \cK(\vF)$. On the other hand,
if $(j,\ell)\in \cK(\vF)$, there might be a larger SEL $(j+1,\ell)$. Thus,
whenever a SEL $(j,\ell)$ is found, it is stored in $(i,k)$ and the search
continues~(line~\ref{line:i-incr}). The algorithm terminates when there are no
more valid SEL candidates and returns the last valid SEL. Note that $\ell$ is
incremented only when there does not exists a larger SEL with the current value
of $\ell$. Thus, for each valid level $j$, if there exists SELs with level $j$,
the algorithm is guaranteed to find the largest such SEL. Moreover, the level is
increased at every possible opportunity. Hence, at the end $(i,k) = \max
\cK(\vF)$.

In the worst case, Alg.~\ref{alg:bot-up} makes at most $3N$ \sat queries.
However, compared to Alg.~\ref{alg:top-down}, the queries are smaller. Moreover,
the computation is incremental and can be aborted with a sub-optimal solution
after execution of line~\ref{line:first-sel} or line~\ref{line:sel}. Note that
at line~\ref{line:first-sel}, $i$ is a regular extension level (i.e., as in
\Avy), and every execution of line~\ref{line:sel} results in a larger SEL.

 \section{Evaluation}
\label{sec:evaluation}
We implemented \Kavy on top of the \Avy Model
Checker\footnote{All code, benchmarks, and results are available at
  \url{https://arieg.bitbucket.io/avy/} }. For line~\ref{line:best_w} of
Alg.~\ref{alg:kavy} we used Alg~\ref{alg:top-down}. We evaluated \Kavy's performance against a
version of \Avy~\cite{DBLP:conf/cav/VizelG14} from the Hardware Model Checking
Competition~2017~\cite{DBLP:conf/fmcad/BiereDH17}, and the \Pdr engine of
ABC~\cite{DBLP:conf/fmcad/EenMB11}. We have used the benchmarks from HWMCC'14,
'15, and '17. Benchmarks that are not solved by any of the solvers are excluded
from the presentation. The experiments were conducted on a cluster running Intel
E5-2683 V4 CPUs at 2.1 GHz with 8GB RAM limit and 30~minutes time limit.
 
\begin{figure}[t]
\centering
\subfloat[ ][Safe HWMCC instances. ]{\label{fig:shiftall}\includegraphics[width=0.48\textwidth]{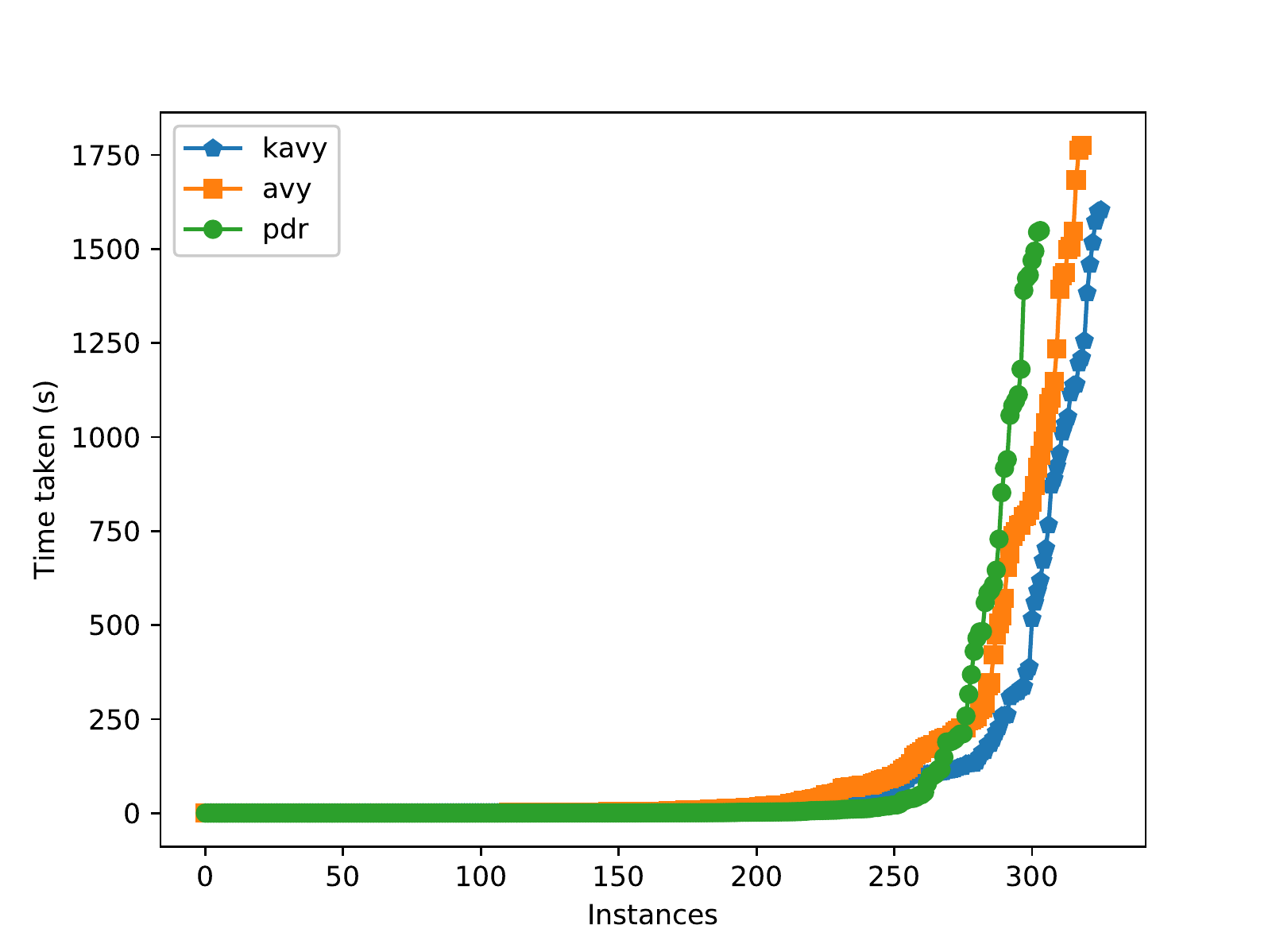}}
\subfloat[ ][\emph{shift} instances. ]{\label{fig:shiftcactus}\includegraphics[width=0.48\textwidth]{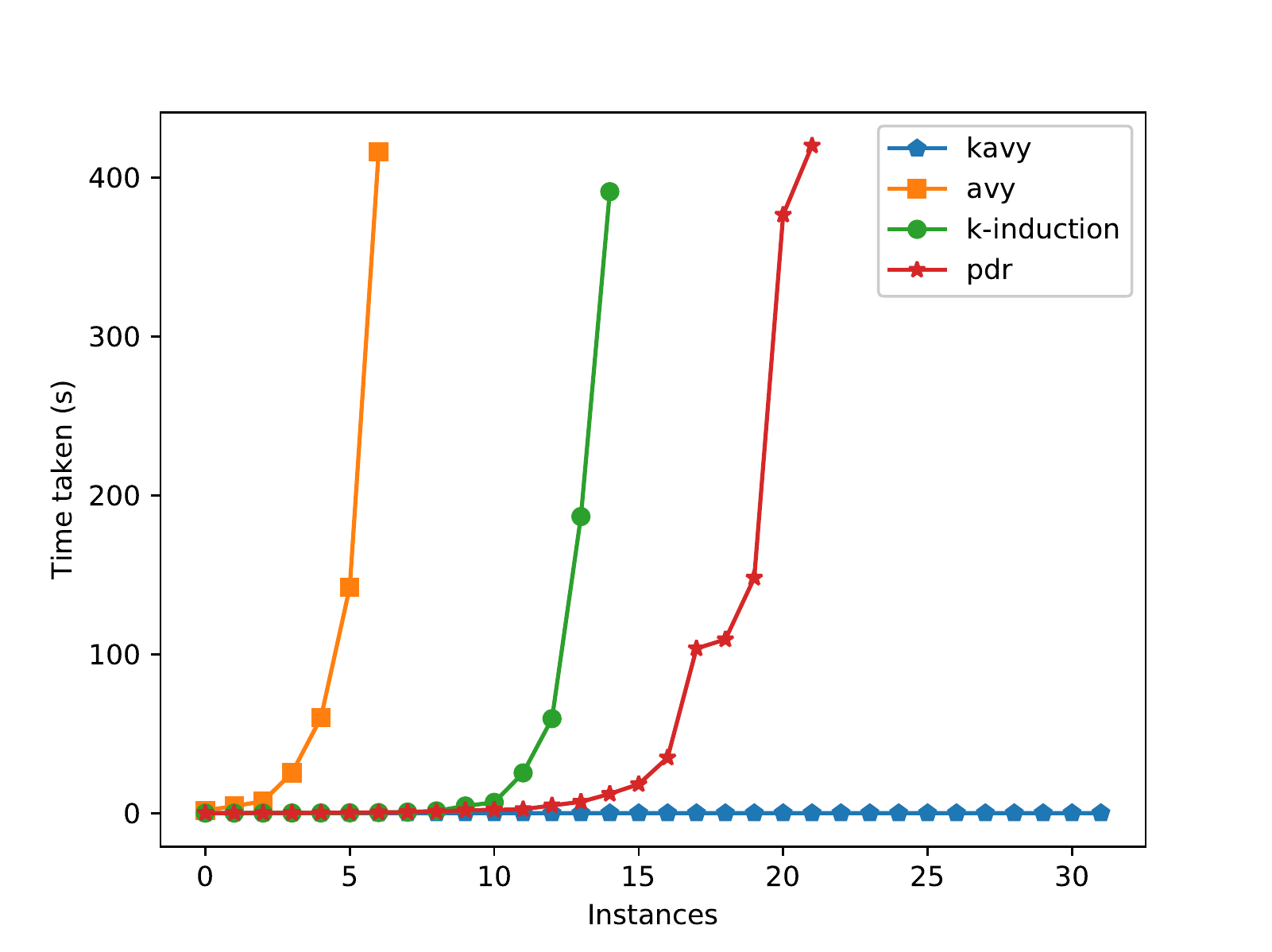}}
\caption{Runtime comparison on SAFE HWMCC instances~\protect\subref{fig:shiftall} and \emph{shift} instances~\protect\subref{fig:shiftcactus}.}
\end{figure}

The results are summarized in Table~\ref{table:fullRun}. The HWMCC has a wide
variety of benchmarks. We aggregate the results based on the competition, and
also benchmark origin (based on the name). Some named categories (e.g.,
\emph{intel}) include benchmarks that have not been included in any competition. The
first column in Table~\ref{table:fullRun} indicates the category. \textbf{Total}
is the number of all available benchmarks, ignoring duplicates. That is, if a
benchmark appeared in multiple categories, it is counted only once. Numbers in
brackets indicate the number of instances that are solved uniquely by the
solver. For example, \Kavy solves 14 instances in \emph{oc8051} that are not
solved by any other solver. The \textsc{VBS} column indicates the \emph{Virtual
  Best Solver} --- the result of running all the three solvers in parallel and
stopping as soon as one solver terminates successfully.

Overall, \Kavy solves more \safe instances than both \Avy and \Pdr, while taking less
time than \Avy (we report time for solved instances, ignoring timeouts). The
\textsc{VBS} column shows that \Kavy is a promising new strategy, significantly
improving overall performance. In the rest of this section, we analyze the
results in more detail, provide detailed run-time comparison between the tools,
and isolate the effect of the new $k$-inductive strategy.

\begin{table}[t]
\centering
\caption{Summary of instances solved by each tool. Timeouts were ignored when computing the time column.}
\label{table:fullRun}
\scalebox{0.8}{\begin{tabular}{|c|c|c|c|c|c|c|c|c|c||c|c|}
\hline
BENCHMARKS	&\multicolumn{3}{|c|}{\Kavy}	&\multicolumn{3}{|c|}{\Avy}&\multicolumn{3}{|c||}{\Pdr}&\multicolumn{2}{|c|}{VBS}\\
\cline{2-12}
 &\safe & \unsafe & time(m) &\safe &\unsafe & time(m)&\safe & \unsafe & time(m)& \safe &\unsafe\\
\hline 
\hline
HWMCC' 17 &$137$~($16$)&$38$ &$499$ &$128$~($3$)&$38$ &$406$ &$109$~($6$)&$40$~($5$) &$174$ &$150$&$44$\\
\hline
HWMCC' 15 &$193$~($4$)&$84$ &$412$ &$191$~($3$)&$92$~($6$) &$597$ &$194$~($16$)&$67$~($12$) &$310$ &$218$&$104$\\
\hline
HWMCC' 14 &$49$&$27$~($1$) &$124$ &$58$~($4$)&$26$ &$258$ &$55$~($6$)&$19$~($2$) &$172$ &$64$&$29$\\
\hline \hline
intel &$32$~($1$)&$9$ &$196$ &$32$~($1$)&$9$ &$218$ &$19$&$5$~($1$) &$40$ &$33$&$10$\\
\hline
6s &$73$~($2$)&$20$ &$157$ &$81$~($4$)&$21$~($1$) &$329$ &$67$~($3$)&$14$ &$51$ &$86$&$21$\\
\hline
nusmv &$13$&$0$ &$5$ &$14$&$0$ &$29$ &$16$~($2$)&$0$ &$38$ &$16$&$0$\\
\hline
bob &$30$&$5$ &$21$ &$30$&$6$~($1$) &$30$ &$30$~($1$)&$8$~($3$) &$32$ &$31$&$9$\\
\hline
pdt &$45$&$1$ &$54$ &$45$~($1$)&$1$ &$57$ &$47$~($3$)&$1$ &$62$ &$49$&$1$\\
\hline
oski &$26$&$89$~($1$) &$174$ &$28$~($2$)&$92$~($4$) &$217$ &$20$&$53$ &$63$ &$28$&$93$\\
\hline
beem &$10$&$1$ &$49$ &$10$&$2$ &$32$ &$20$~($8$)&$7$~($5$) &$133$ &$20$&$7$\\
\hline
oc8051 &$34$~($14$)&$0$ &$286$ &$20$&$0$ &$99$ &$6$~($1$)&$1$~($1$) &$77$ &$35$&$1$\\
\hline
power &$4$&$0$ &$25$ &$3$&$0$ &$3$ &$8$~($4$)&$0$ &$31$ &$8$&$0$\\
\hline
shift &$5$~($2$)&$0$ & $1$ &$1$&$0$ &$18$ &$3$&$0$ & $1$ &$5$&$0$\\
\hline
necla &$5$&$0$ &$4$ &$7$~($1$)&$0$ & $1$ &$5$~($1$)&$0$ &$4$ &$8$&$0$\\
\hline
prodcell &$0$&$0$ &$0$ &$0$&$1$ &$28$ &$0$&$4$~($3$) &$2$ &$0$&$4$\\
\hline
bc57 &$0$&$0$ &$0$ &$0$&$0$ &$0$ &$0$&$4$~($4$) &$9$ &$0$&$4$\\
\hline \hline
\textbf{Total} &$326$~($19$)&$141$~($1$) &$957$ &$319$~($8$)&$148$~($6$) &$1041$ &$304$~($25$)&$117$~($17$) &$567$ &$370$&$167$\\
\hline
\end{tabular}
}
\end{table}

\begin{figure}[t]
\centering
\subfloat[ ][\Kavy vs \Avy]{\label{fig:kavy-avy}\includegraphics[width=0.3\textwidth]{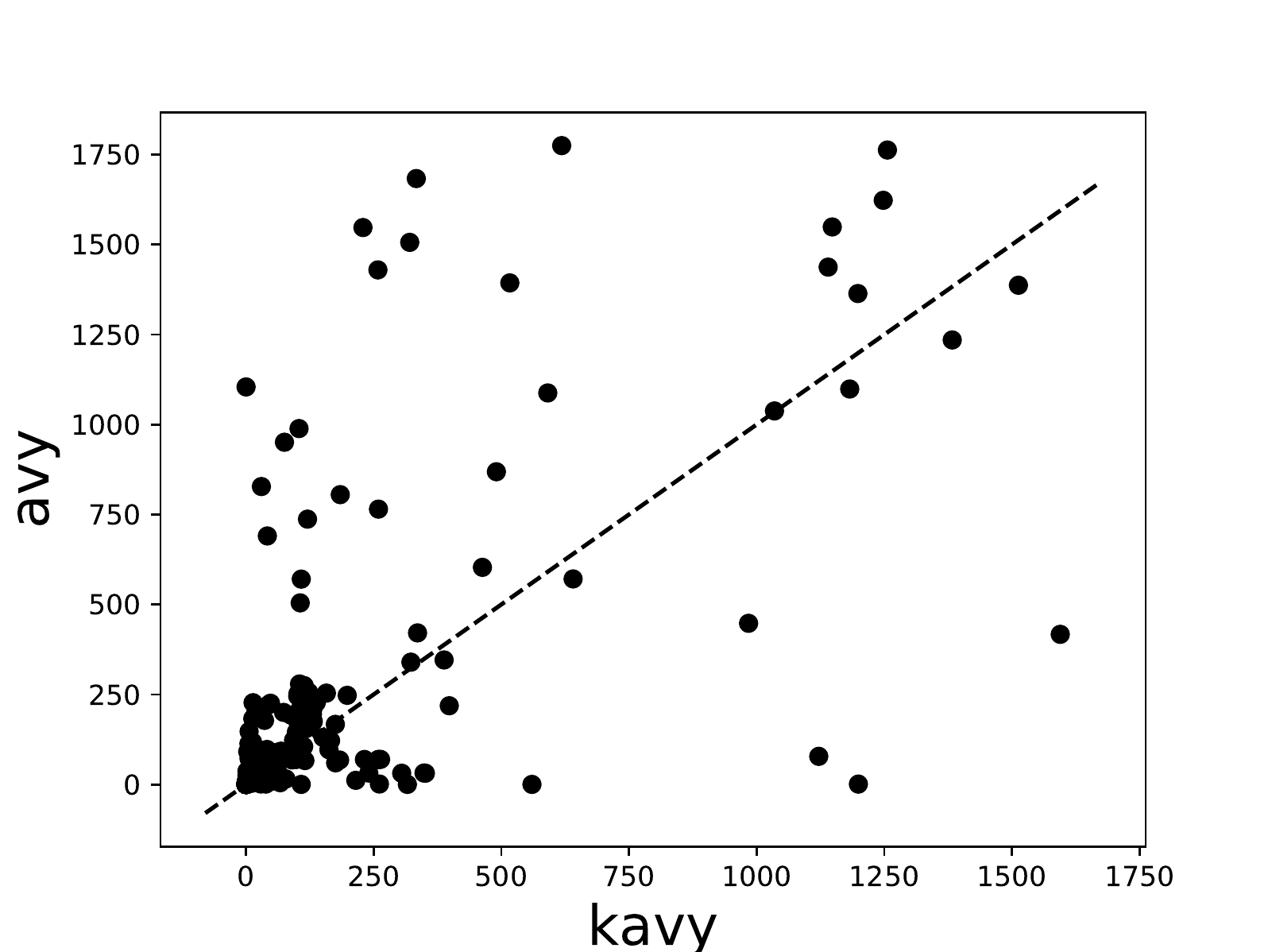}}
\subfloat[ ][\Kavy vs \Pdr]{\label{fig:kavy-abc}\includegraphics[width=0.3\textwidth]{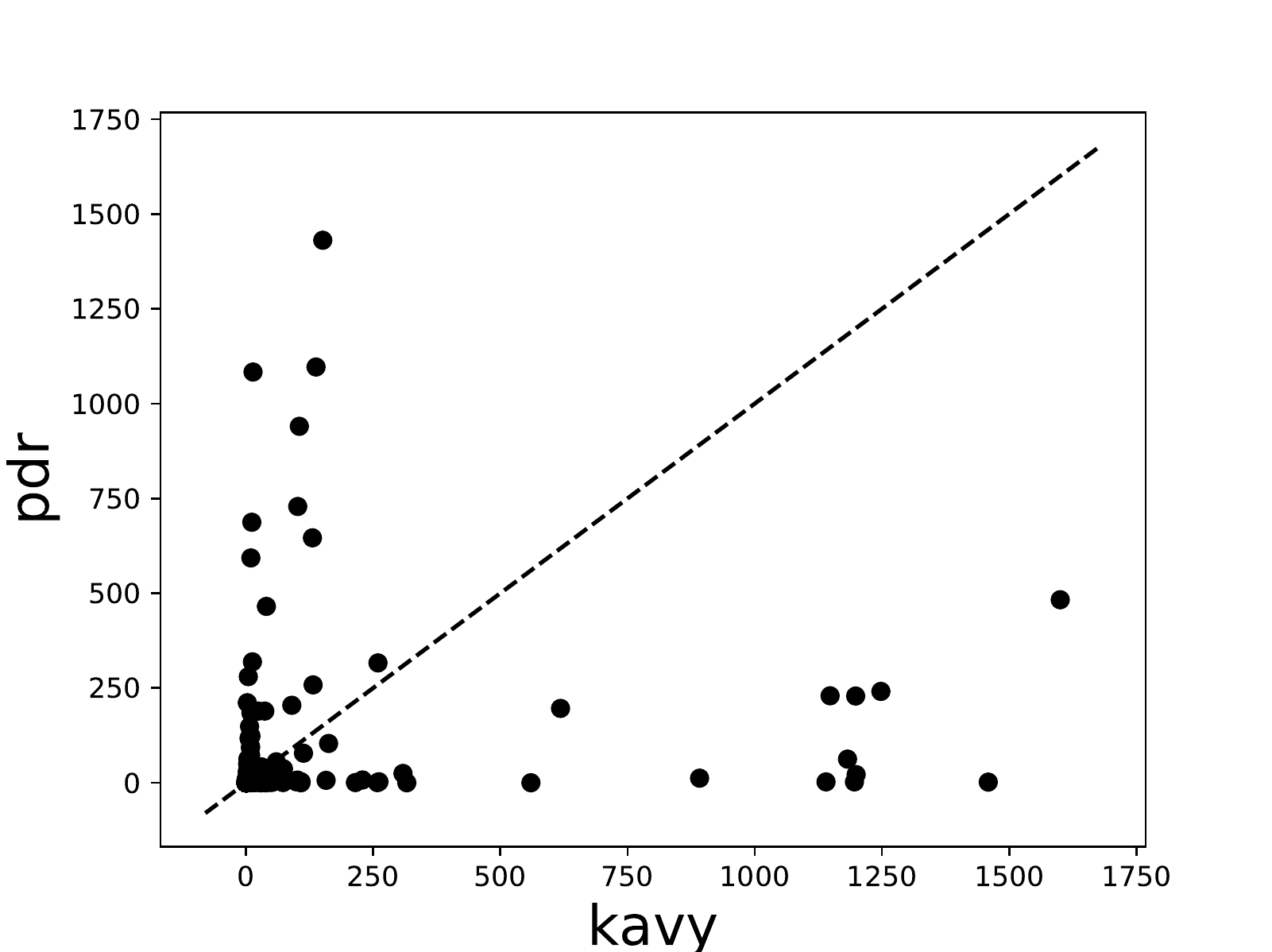}}
\subfloat[ ][\Kavy vs \Kavyvanilla]{\label{fig:kavy-vanilla}\includegraphics[width=0.3\textwidth]{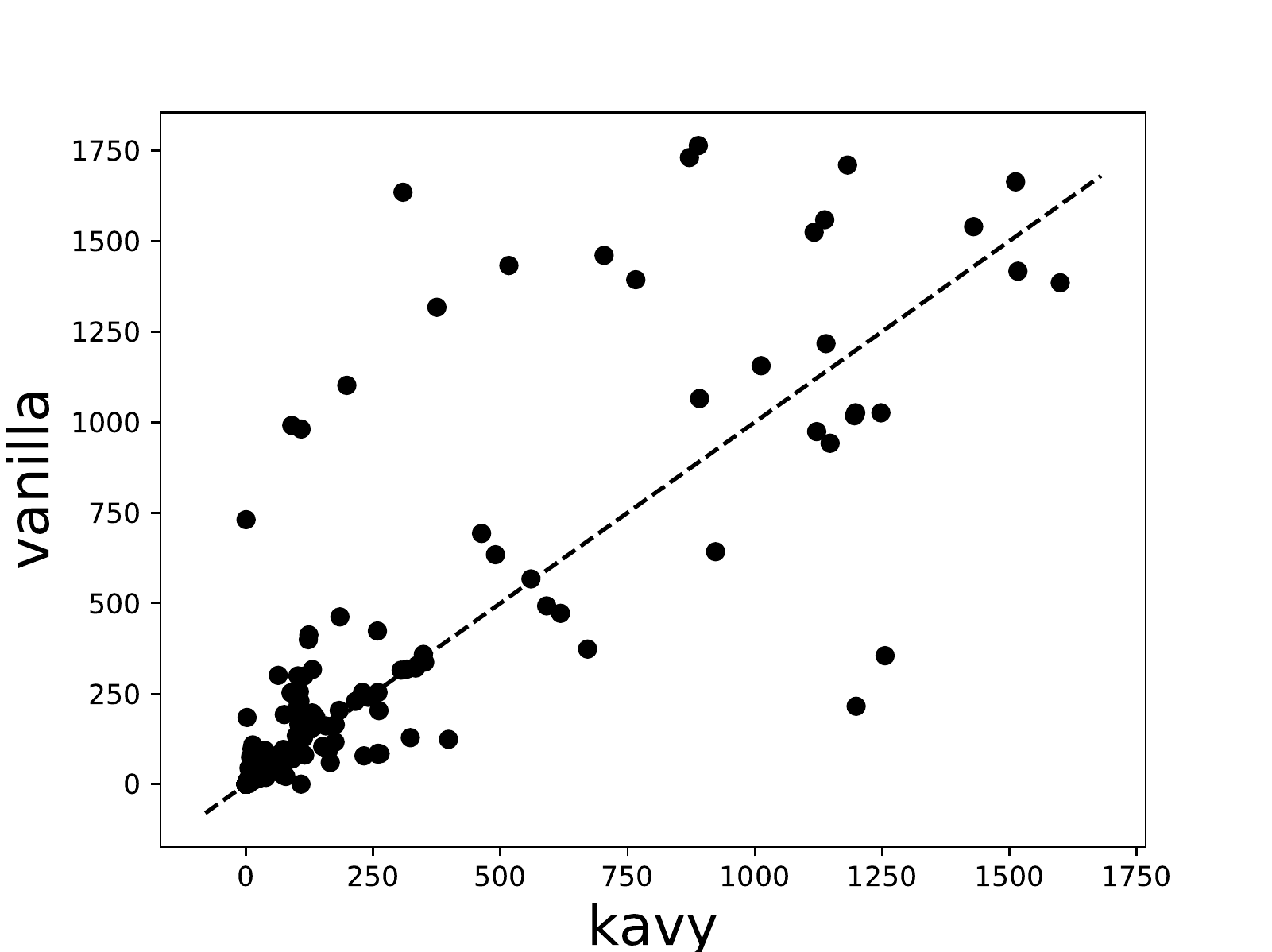}}
\label{fig:running-times}
\subfloat[ ][\Kavy vs \Avy]{\label{fig:kavy-avy-depth}\includegraphics[width=0.3\textwidth]{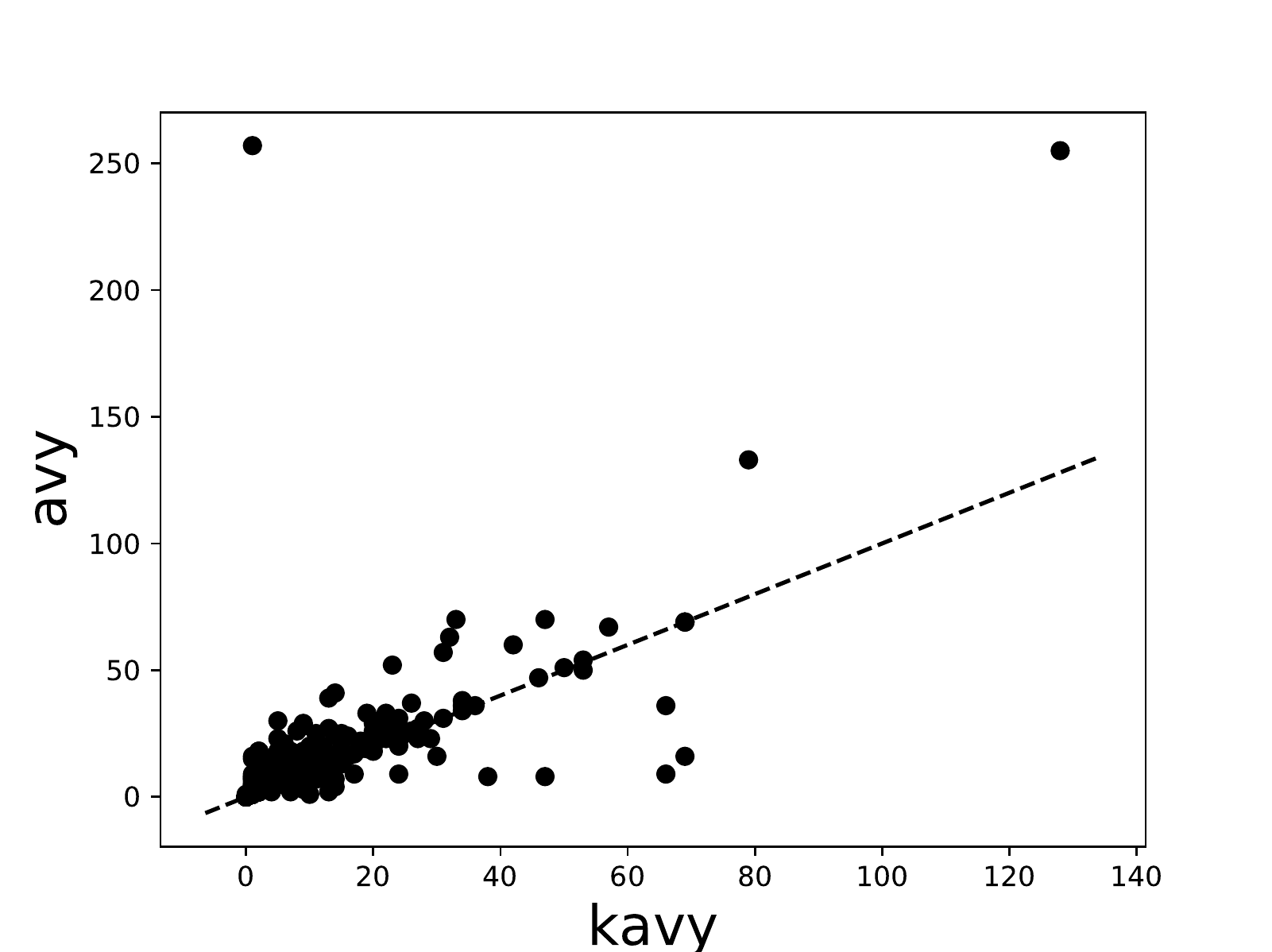}}
\subfloat[ ][\Kavy vs \Pdr]{\label{fig:kavy-abc-depth}\includegraphics[width=0.3\textwidth]{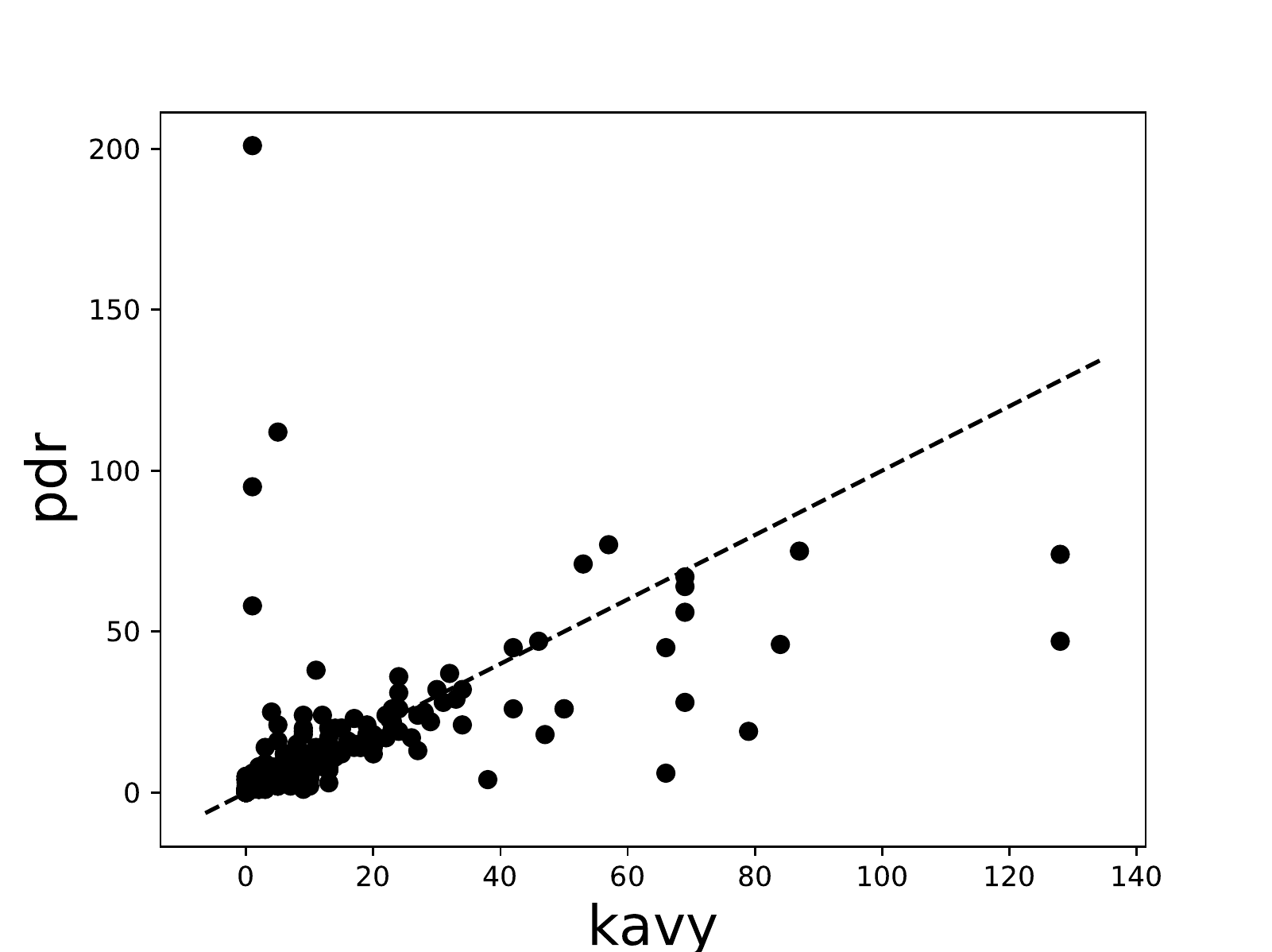}}
\subfloat[ ][\Kavy vs \Kavyvanilla]{\label{fig:kavy-vanilla-depth}\includegraphics[width=0.3\textwidth]{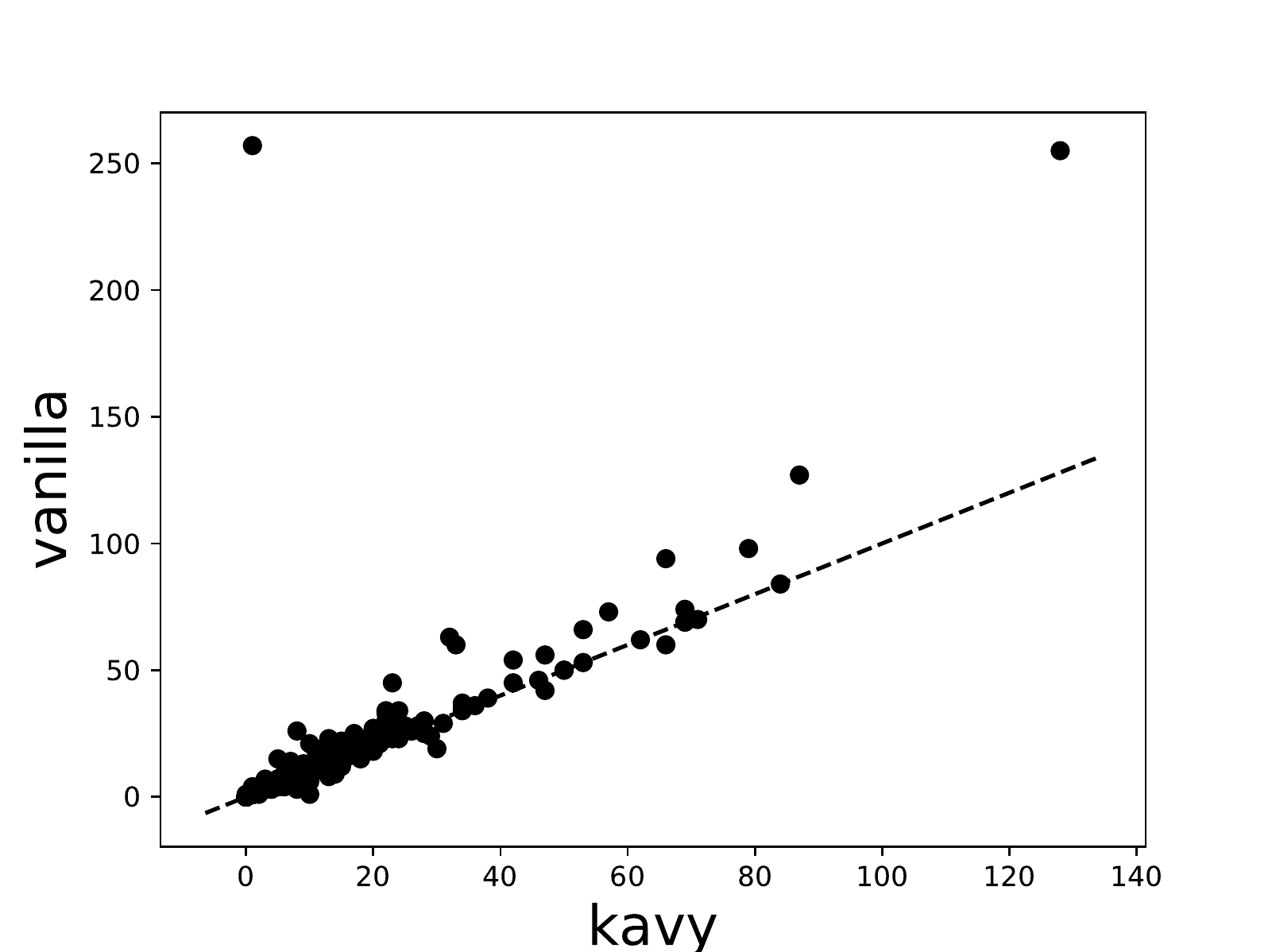}}
\label{fig:depths}
\caption{Comparing running time~(\protect\subref{fig:kavy-avy},~\protect\subref{fig:kavy-abc},~\protect\subref{fig:kavy-vanilla}) and depth of convergence~(\protect\subref{fig:kavy-avy-depth},~\protect\subref{fig:kavy-abc-depth},~\protect\subref{fig:kavy-vanilla-depth}) of \Avy, \Pdr and \Kavyvanilla with \Kavy. \Kavy is shown on the x-axis. Points above the diagonal are better for \Kavy. Only those instances that have been solved by both solvers are shown in each plot. }
\end{figure}

To compare the running time, we present scatter plots comparing \Kavy and \Avy
(Fig.~\ref{fig:kavy-avy}), and \Kavy and \Pdr (Fig.~\ref{fig:kavy-abc}). In both
figures, \Kavy is at the bottom. Points above the diagonal are better for \Kavy.
Compared to \Avy, whenever an instance is solved by both solvers, \Kavy is often
faster, sometimes by orders of magnitude. Compared to \Pdr, \Kavy and \Pdr
perform well on very different instances. This is similar to the observation
made by the authors of the original paper that presented
\Avy~\cite{DBLP:conf/cav/VizelG14}. Another indicator of performance is the
depth of convergence. This is summarized in Fig.~\ref{fig:kavy-avy-depth} and
Fig.~\ref{fig:kavy-abc-depth}. \Kavy often converges much sooner than \Avy. The
comparison with \Pdr is less clear which is consistent with the difference in
performance between the two. To get the whole picture, Fig.~\ref{fig:shiftall} presents a cactus plot that compares the running times of the algorithms on all these benchmarks.

To isolate the effects of $k$-induction, we compare \Kavy to a version of \Kavy
with $k$-induction disabled, which we call \Kavyvanilla. Conceptually,
\Kavyvanilla is similar to \Avy since it extends the trace using a $1$-inductive
extension trace, but its implementation is based on \Kavy. The results for the
running time and the depth of convergence are shown in
Fig.~\ref{fig:kavy-vanilla} and Fig.~\ref{fig:kavy-vanilla-depth}, respectively.
The results are very clear --- using strong extension traces significantly
improves performance and has non-negligible affect on depth of convergence.

Finally, we discovered one family of benchmarks, called shift, on which \Kavy
performs orders of magnitude better than all other techniques. The benchmarks
come from encoding bit-vector decision problem into
circuits~\cite{DBLP:journals/mst/KovasznaiFB16,DBLP:conf/fmcad/VizelNM17}. The
shift family corresponds to deciding satisfiability of $(x+y)=(x << 1)$ for two
bit-vecors $x$ and $y$. The family is parameterized by bit-width. The property
is $k$-inductive, where $k$ is the bit-width of $x$. The results of running
\Avy, \Pdr, $k$-induction\footnote{We used the $k$-induction engine \texttt{ind}
  in \textsc{Abc}~\cite{DBLP:conf/cav/BraytonM10}.}, and \Kavy are shown in
Fig.~\ref{fig:shiftcactus}. Except for \Kavy, all techniques exhibit exponential
behavior in the bit-width, while \Kavy remains constant. Deeper analysis
indicates that \Kavy finds a small inductive invariant while exploring just
two steps in the execution of the circuit. At the same time, neither inductive
generalization nor $k$-induction alone are able to consistently find the same
invariant quickly.

 \section{Conclusion}
\label{sec:conclusion}
In this paper, we present \Kavy --- an SMC algorithm that effectively uses
$k$-inductive reasoning to guide interpolation and inductive generalization.
\Kavy searches both for a good inductive strengthening and for the most
effective induction depth $k$. We have implemented \Kavy on top of \Avy Model
Checker. The experimental results on HWMCC instances show that our approach is
effective. 

The search for the maximal SEL is an overhead in \Kavy. There could be
benchmarks in which this overhead outweighs its benefits. However, we have not
come across such benchmarks so far. In such cases, \Kavy can choose to settle
for a sub-optimal SEL as mentioned in section~\ref{sec:kiew}. Deciding when and
how much to settle for remains a challenge.

 \vspace{-10pt}
\subsubsection{Acknowledgements}
We thank the anonymous reviewers and Oded Padon for their thorough review and
insightful comments. This research was enabled in part by support provided by
Compute Ontario (\url{https://computeontario.ca/}), Compute Canada
(\url{https://www.computecanada.ca/}) and the grants from Natural Sciences and
Enginerring Research Council Canada.

\bibliographystyle{plain}
\bibliography{ref}

\end{document}